\documentclass[acmlarge,screen]{acmart}

\newif\ifarxiv 
\arxivtrue
\ifarxiv\makeatletter
\@ACM@nonacmtrue
\@ACM@printacmreffalse
\makeatother\fi

\usepackage[activeacute,english]{babel} 
\usepackage[utf8]{inputenc}

\usepackage{amsmath,amsthm,amsfonts,amscd}
\usepackage{multirow,booktabs}
\usepackage{faktor}
\usepackage{mathrsfs}
\usepackage{kantlipsum}
\usepackage{mathtools}
\usepackage{cleveref}

\usepackage{xspace}
\usepackage{multicol}
\usepackage[linesnumbered,vlined,boxruled,boxed,algo2e]{algorithm2e}
\usepackage{enumitem}
\usepackage{tikz}

\usepackage{bm}

\numberwithin{equation}{section} 
\usepackage{tikz-cd} 
\tikzset{
  symbol/.style={
    draw=none,
    every to/.append style={
      edge node={node [sloped, allow upside down, auto=false]{$#1$}}}
  }
}
\usetikzlibrary{patterns,calc}
\usepackage{float} 
\usepackage[labelfont=bf,labelsep=space]{caption} 
\usepackage{empheq} 

\usepackage{subcaption}
\usepackage[textsize=tiny]{todonotes}

\newtheorem{lem}{Lemma}[section]
\newtheorem{prop}[lem]{Proposition}
\newtheorem{theo}[lem]{Theorem}
\newtheorem{cor}[lem]{Corollary}

\theoremstyle{definition}
\newtheorem{defi}[lem]{Definition}
\theoremstyle{remark}

\newtheorem{exam}[lem]{Example}

\newtheorem{remark}[lem]{Remark}
\newtheorem{assumption}[lem]{Assumption}

\newcommand{\eproof}{\hfill\qed}

\def\mcP{\mathcal{P}}

\def\mcZ{\mathcal{Z}}

\def\bbE{\mathbb{E}}

\def\bbR{\mathbb{R}}
\def\bbZ{\mathbb{Z}}
\def\bbN{\mathbb{N}}

\def\bbC{\mathbb{C}}
\def\bbP{\mathbb{P}}

\def\bbD{\mathbb{D}}
\def\fkf{\mathfrak{f}}

\def\fkc{\mathfrak{c}}

\def\fkx{\mathfrak{x}}
\def\fky{\mathfrak{y}}
\def\fkz{\mathfrak{z}}
\def\sgn{\mathsf{sgn}}

\def\NN{\mathbb{N}\xspace}

\def\RR{\mathbb{R}\xspace}

\DeclareMathOperator{\rem}{rem}

\def\Pd{\mcP_{d}}

\def\Oh{\mathcal{O}}

\DeclareMathOperator{\cond}{\texttt{C}}
\DeclareMathOperator{\condR}{\texttt{C}_{\bbR}}

\def\uR{^{\mathbb{R}}}
\def\uC{^{\mathbb{C}}}
\def\enumber{\mathrm{e}}
\def\dist{\mathrm{dist}}

\DeclarePairedDelimiter\ceil{\lceil}{\rceil}

\DeclarePairedDelimiter\abs{\lvert}{\rvert}
\DeclarePairedDelimiter\norm{\lVert}{\rVert}

\newcommand{\OO}{\ensuremath{\mathcal{O}}\xspace}
\newcommand{\sO}{\ensuremath{\widetilde{\mathcal{O}}}\xspace}
\newcommand{\OB}{\ensuremath{\mathcal{O}_B}\xspace}
\newcommand{\sOB}{\ensuremath{\widetilde{\mathcal{O}}_B}\xspace}

\newcommand{\sturm}{\textsc{sturm}\xspace}
\newcommand{\Descartes}{\textsc{descartes}\xspace}
\newcommand{\descartes}{\textsc{descartes}\xspace}
\newcommand{\anewdsc}{\textsc{aNewDsc}\xspace}
\newcommand{\jssparse}{\textsc{JS-sparse}\xspace}

\newcommand{\uObD}{{\overline{\mathcal{D}}}\xspace}
\newcommand{\dObD}{{\underline{\mathcal{D}}}\xspace}
\newcommand{\ObL}{\mathcal{L}\xspace}
\newcommand{\ObA}{\mathcal{A}\xspace}

\newcommand{\wid}{\mathtt{wid}\xspace}
\newcommand{\xmid}{\mathtt{mid}\xspace}
\newcommand{\img}{\mathtt{i}\xspace}

\newcommand{\ST}{\mathtt{ST}\xspace}

\newcommand{\var}{\textsc{var}\xspace}
\newcommand{\func}[1]{\textsc{#1}\xspace}

\makeatletter
\let\original@algocf@latexcaption\algocf@latexcaption
\long\def\algocf@latexcaption#1[#2]{%
  \@ifundefined{NR@gettitle}{%
    \def\@currentlabelname{#2}%
  }{%
    \NR@gettitle{#2}%
  }%
  \original@algocf@latexcaption{#1}[{#2}]%
}
\makeatother

\usepackage{comment}

\AtBeginDocument{%
  \providecommand\BibTeX{{%
    \normalfont B\kern-0.5em{\scshape i\kern-0.25em b}\kern-0.8em\TeX}}}

\setcopyright{acmcopyright}
\copyrightyear{2018}
\acmYear{2018}
\acmDOI{XXXXXXX.XXXXXXX}

\begin{document}

\title{Beyond Worst-Case Analysis for Symbolic Computation: Root Isolation Algorithms}


\author{Alperen A. Ergür}
\authornote{All authors contributed equally to this research.
The paper is an extended version of a conference paper presented in ISSAC 2022 \cite{etct-bwc-issac-22}.}
\email{alperen.ergur@utsa.edu}
\orcid{0000-0002-2340-6551}
\authornotemark[1]
\affiliation{%
  \institution{The University of Texas at San Antonio}
  \department{Department of Computer Science and Department of Mathematics}
  \streetaddress{One UTSA Circle}
  \city{San Antonio}
  \state{TX}
  \postcode{78249}
  \country{USA}
}

\author{Josué Tonelli-Cueto}
\authornotemark[1]
\email{josue.tonelli.cueto@bizkaia.eu}
\orcid{0000-0002-2904-1215}
\affiliation{%
  \institution{Johns Hopkins University}
  \department{Department of Applied Mathematics and Statistics}
  \streetaddress{3400 North Charles Street}
  \city{Baltimore}
  \state{MD}
  \postcode{21218}
  \country{USA}
}

\author{Elias Tsigaridas}
\authornotemark[1]
\email{elias.tsigaridas@inria.fr}
\affiliation{%
  \institution{Inria Paris \& Sorbonne University}
  \department{IMJ-PRG}
  \city{Paris}
  \country{France}}

\renewcommand{\shortauthors}{A.A. Ergür, J. Tonelli-Cueto and E. Tsigaridas}

\begin{abstract}
We introduce beyond-worst-case analysis into symbolic computation. This is an extensive field which almost entirely relies on worst-case bit complexity, and we start from a basic problem in the field: isolating the real roots of univariate polynomials. This is a fundamental problem in symbolic computation and it is arguably one of the most basic problems in computational mathematics. The problem has a long history decorated with numerous ingenious algorithms and furnishes an active area of research. However, most available results in literature either focus on worst-case analysis in the bit complexity model or simply provide experimental benchmarking without any theoretical justifications of the observed results. We aim to address the discrepancy between practical performance of root isolation algorithms and prescriptions of worst-case complexity theory: We develop a smoothed analysis framework for polynomials with integer coefficients to bridge this gap.  
We demonstrate (quasi-)linear (expected and smoothed) complexity bounds
for \Descartes algorithm, that is one most well know symbolic algorithms for isolating the real roots of univariate polynomials with integer coefficients. 
Our results explain the surprising efficiency of Descartes solver in comparison to sophisticated algorithms that have superior worst-case complexity.
We also analyse the \sturm solver,  \anewdsc a symbolic-numeric algorithm that combines \Descartes with Newton operator,
and a symbolic algorithm for sparse polynomials.
\end{abstract}



\received{????}
\received[revised]{????}
\received[accepted]{????}

\maketitle

\newpage
\tableofcontents
\newpage

\section{Introduction}

The complementary influence between design and analysis of algorithms has transformative implications on both domains. On the one hand, surprisingly efficient algorithms, such as the simplex algorithm, reshape the landscape of complexity analysis frameworks. On the other hand, the identification of fundamental complexity parameters has the potential to transform the algorithm development; the preconditioned conjugate gradient algorithm is a case in point. This interplay between complexity analysis frameworks and algorithmic design  represents a dynamic and vibrant area of contemporary research in discrete computation \cite{roughgarden2021book,downeyfellows2013} with roots in the early days of complexity theory~\cite[Ch.~18]{arorabarak2009}. This line of thought already demonstrated remarkable success starting from the pioneering work of Spielman and Teng on linear programming \cite{sp-smooth-jacm}, continued with remarkable works on local search algorithms in discrete optimization \cite{roughgarden2021book}[Chapters 13 and 15], and more recently in other fields such as online algorithms \cite{haghtalab2020smoothed} and statistical learning \cite{diakonikolas2023algorithmic}.  All these efforts fall under the umbrella of the framework of \emph{beyond worst case analysis of algorithms}.

In the (specific) domain of numerical algorithms, condition numbers proved to be the fundamental notion connecting design and complexity analysis of algorithms. On the one hand, condition numbers provide a means to elucidate the success of specific numerical algorithms, and on the other hand, are the pivotal complexity parameters guiding the development of novel algorithms. This was already noticed by Turing~\cite{turing1948} in his efforts to explain the practical efficacy of Gaussian elimination as documented by Wilkinson~\cite{wilkinson1971}. This tight connection of theoretical and practical aspects of numerical computation resulted in "the ability to compute quantities that are typically uncomputable from an analytical point of view, and do it with lightning speed", quoting Trefethen \cite{trefethen1992definition}.

Motivated by the success of  beyond worst-case analysis in general
and the success of condition numbers in numerical algorithms in particular, we embark on an endeavor to introduce such algorithmic analysis tools into the domain of symbolic computation.  To the best of our knowledge, this expansive field has predominantly relied upon worst-case bit complexity for analysis of algorithms. More precisely, in this paper we pursue two ideas simultaneously: (1)~develop a theory of condition numbers as a basic parameter to understand the behaviour of symbolic algorithms, and (2)~develop data models on discrete, that is integer input, that captures the problem instances in symbolic computation.

Our overarching aim is to enrich symbolic computation with ideas from beyond worst-case analysis and numerical computation. So we naturally start from the most basic and fundamental questions in this field: We work on delineating the performance of algorithms for computing
the roots of univariate polynomials. This is a singularly important problem with whole range of applications in computer science and engineering. It is
extensively studied from theoretical and practical perspectives for decades and
it is still a very active area of research
\cite{McP-book-13,Pan-survey-97,ept-crc-2012,pan-bb-soda-2024,moroz2021,ImbPan-issac-2020,pan-survey-2022,MorImb-rf-23}.
Our main focus is on the \emph{real root isolation} problem:
given a univariate polynomial with integer coefficients, our goal is to compute intervals with
rational endpoints that contain only one real root of the polynomial and each
real root is contained in an interval. Besides its countless direct
applications, this problem is omnipresent in symbolic computation; it is a crucial subroutine for elimination-based multivariate polynomial systems solvers, see e.g., ~\cite{ept-crc-2012}.

Despite the ubiquity of (real) root isolation in engineering and its relatively
long history in theoretical computer science, the state-of-the-art complexity
analysis falls short of providing guidance for practical computations. Pan's
algorithm~\cite{Pan02jsc}, which finds, that is approximates, all the complex roots and not just the real ones, has the best worst-case complexity since nearly two
decades; it is colloquially referred to as the "optimal" algorithm. However,
Pan's algorithm is rather sophisticated and has only, to our knowledge, a prototype implementation
in PARI/GP \cite{PARI2}.
In contrast, other algorithms with inferior worst-case complexity estimates
have excellent practical performance, e.g.,~
\cite{krs-for-real-16,htzekm-solve-09,et-tcs-2007,ImbPan-issac-2020}.
The algorithms that are used in practice,
even though they achieve disappointing worst case (bit) complexity bounds, are conceptually simpler and, surprisingly, they outperform the rivals with superior worst-case bounds by several orders of magnitude \cite{t-slv-16,RouZim:solve:03,htzekm-solve-09}.  In our view, this lasting
discrepancy between theoretical complexity analyses and practical performance
is related to the insistence on using the worst-case framework in the symbolic
computation community besides a few exceptions, e.g.~\cite{egt-issac-2010,et-tcs-2007,PanTsi-refine-2013}. 


Despite the importance of root isolation and its extensive literature, bearing the aforementioned few exceptions, there remains a big discrepancy between theoretical analysis and practice of solving univariate polynomials. Basically, symbolic computation literature lacks appropriate randomness models and technical tools to perform beyond worst-case analysis. Our approach addresses this gap.
 
We introduce tools that allow us to demonstrate how average/smoothed analysis frameworks can help to
 predict the practical performance of symbolic (real) root isolation
  algorithms. In particular, we show that in our discrete random  model the 
\descartes solver, a solver commonly used in practice, has quasi-linear bit complexity in the input size. This provides an explanation for the excellent practical performance of \descartes: See \Cref{subsec:results1} for a simple statement and \Cref{subsec:results2} for the full technical statement.
Besides \descartes, we consider \sturm solver (\Cref{sec:Sturm})
that is based on Sturm' sequences. Our average and smoothed analysis bounds 
are worse than the one of \descartes by an order of magnitude.
This provides the first theoretical explanation of the superiority of \descartes over \sturm that is commonly seen in practice. 
In addition, we analyze a hybrid symbolic/numeric solver, \anewdsc, 
(\Cref{sec:aNewDsc})
that combines Descartes' rule of signs with Newton operators; its bounds are similar to \descartes. 
Finally, we consider \jssparse solver by Jindal and Sagraloff \cite{JinSag-sparse-17}, that isolates the real roots of univariate polynomials in the sparse encoding (\Cref{sec:jindal-sagraloff}). We are not aware of any other analysis, except worst case, of a sparse solver. 

To justify our main focus on \descartes solver
we emphasize that is the symbolic algorithm 
commonly used in practice because of its simplicity and efficiency. 
Furthermore, from the theoretical point of view, 
it is the algorithm that requires the widest arsenal of tools
for its beyond worst-case analysis: We can analyze the other solvers
using a (suitable modified) subset of the tools that we employ for \descartes, not necessarily the same for all of them.



\subsection{Warm-up: A simple form of the main results}
\label{subsec:results1}

The main complexity parameters for univariate polynomials with integer (or
rational) coefficients is the degree $d$ and the bitsize $\tau$; the latter
refers to the maximum bitsize of the coefficients. We aim for a data model that
resembles a ``typical'' polynomial with exact coefficients. The first natural
candidate is the following: fix a bitsize $\tau$, let
$\fkc_0,\fkc_1,\ldots,\fkc_d$
be independent copies of a uniformly distributed integer in
$[-2^{\tau},2^{\tau}] \cap \mathbb{Z}$, and consider the polynomial
$
  \fkf = \sum\nolimits_{i=0}^d \fkc_i X^{i}
$,
which we call the \emph{uniform random bit polynomial with bitsize $\tau(\fkf)$}.  For this polynomial, we prove the following result(s):
\begin{theo}\label{theo:main1}
	Let $\fkf$ be a uniform random bit polynomial, 
	of degree $d$ and bit size	$\tau := \tau(\fkf)$.
	We can isolate the real roots of $\fkf$ in $I=[-1, 1]$ using 
	\begin{itemize}
		\item \descartes in expected time $\sOB(d^2 + d \, \tau)$ (\Cref{thm:Descartes-complexity-exp}),
		\item \sturm in expected time $\sOB(d^2 \tau)$ (\Cref{thm:Sturm-complexity-exp}), and 
		\item \anewdsc in expected time $\sOB(d^2 + d \, \tau)$ (\Cref{thm:aNewDsc-complexity-exp}).
	\end{itemize}
	
	If $\fkf$ is a sparse polynomial having at most $M$ terms, then, using 
	the \jssparse algorithm, we	can isolate its real roots in 
	expected time $\sOB\left(|M|^{12} \, \tau^2  \, \log^{3}{d} \right)$ (\Cref{thm:jssparse-complexity-exp}).
\end{theo}
We use $\OO$, resp. $\OB$, to denote the arithmetic, resp.  bit,
complexity and $\sO$, resp. $\sOB$, when we ignore the (poly-)logarithmic
factors of $d$. As we will momentarily explain the expected time complexity of \descartes solver in this simple model is better by a factor of $d$ than the record worst-case complexity bound of Pan's algorithm, provided that $d$ is comparable with $\tau$.

\subsection{A brief overview of (real) root isolation algorithms}
\label{sec:synopsis}
The bibliography on the problem of root finding of univariate polynomials is vast and our presentation of the relevant literature just represents the tip of the iceberg. We encourage the curious reader to consult the bibliography of the cited references.

We can (roughly) characterize the various algorithms for (real) root isolation
as numerical or symbolic algorithms;  the recent years there are also efforts
to combine the best of the two worlds.
The numerical algorithms are, in almost all the cases, iterative algorithms that
approximate all the roots (real and complex) of a polynomial up to any desired
precision. Their main common tool is (a variant of) a Newton operator;
with only a few exceptions that  use the root-squaring operator of Dandelin, Lobachevsky, and Gräffe. The algorithm with the best worst-case complexity,  due to Pan~\cite{Pan02jsc}, employs Sch\"onhage's splitting circle divide-and-conquer technique~\cite{Sch82}. 
It recursively factors the polynomial until we obtain linear
factors that approximate, up to any desired precision, all the roots of the
polynomial and it has nearly optimal arithmetic complexity. We can turn this
algorithm, and also any other numerical algorithm, to an exact one, by
approximating the roots up to the separation bound; that is
the minimum distance between the roots. In this way, Pan obtained the record
worst case bit complexity bound $\sOB(d^2\tau)$ for a degree $d$ polynomial
with maximum coefficient bitsize $\tau$ \cite{Pan02jsc}; see also
\cite{Kirrinis-solve-98,msw-apan-15,becker2018near}.
Besides the algorithms already mentioned, there are also several seemingly
practically efficient numerical algorithms, e.g., \textsc{mpsolve}
\cite{mpsolve-theory} and \textsf{eigensolve} \cite{eigensolve}, that lack
convergence guarantees and/or precise bit complexity estimates.

Regarding symbolic algorithms, the majority are subdivision-based
and they mimic binary search. Given an initial interval that contains all (or
some) of the real roots of a square-free univariate polynomial with integer coefficients, they repeatedly subdivide it until we obtain intervals
containing zero or one real root. 
Prominent representatives of this approach are \sturm and \descartes. \sturm
depends on Sturm sequences to count \emph{exactly} the number of distinct roots
in an interval, even when the polynomial is not square-free. Its complexity is
$\sOB(d^4 \tau^2)$ \cite{Dav:TR:85,Yap:SturmBound:05} and it is not so
efficient in practice; the
bottleneck seems to be the high cost of  computing the Sturm sequence.
\descartes is
based on Descartes' rule of signs to bound the number of real roots of a
polynomial in an interval. Its worst case complexity is
$\sOB(d^4 \tau^2)$ \cite{ESY:descartes}. Even though its worst case bound is
similar to \sturm, the \descartes solver has excellent practical performance
and it can routinely solve polynomials of degree several thousands
\cite{RouZim:solve:03,JohKraLynRicRus-issac-06,t-slv-16,htzekm-solve-09}.
There are also other algorithms based on the continued fraction expansion of
the real numbers \cite{sharma-tcs-2008,et-tcs-2007}
and on point-wise evaluation \cite{BurrKrahmer-eval-12,sy-simple-11}.

Let us also mention a variant of \descartes
\cite{ekkmsw-bitstream-05}, where we assume an oracle that for each coefficient of the
polynomial returns an approximation to any absolute error.
In this setting, by incorporating several tools from numerical algorithms, one obtains an improved variant of \descartes \cite{sm-anewdsc, krs-for-real-16}. For recent progress of this algorithm we refer to \cite{ImbPan-issac-2020}.  There is also a subdivision algorithm  \cite{becker2018near}
that improves upon earlier work 
\cite{Pan-Weyl-00} with very good worst-case complexity bounds. 
Finally, let us mention that there are also root finding algorithms based on the condition number and efficient floating point computations
\cite{MorImb-rf-23,moroz2021} and also algorithms that consider the black box model~\cite{pan-bb-soda-2024}.

\subsection{Statement of main results in full detail}\label{subsec:results2}

We develop a general model of randomness that provides the framework of smoothed analysis
for polynomials with integer coefficients.

\begin{defi}
Let $d\in\bbN$. A \emph{random bit polynomial with degree $d$} is a random
polynomial
$
\fkf:=\sum\nolimits_{i=0}^d \fkc_i X^{i} ,
$
where the $\fkc_i$ are independent discrete random variables with values in
$\bbZ$.
Then,
\begin{enumerate}
  \item the \emph{bitsize of $\fkf$}, $\tau(\fkf)$, is the minimum integer
    $\tau$ such that, for all $i \in \{ 0,1,2,\ldots,d \}$,
    $
    \bbP(|\fkc_i| \leq 2^\tau)=1.
    $
  \item the \emph{weight of $\fkf$}, $w(\fkf)$, is the maximum probability that
    $\fkc_0$ and $\fkc_d$ can take a value, that is 
    \[
	    w(\fkf):=\max\{\bbP(\fkc_i=a) \mid  i\in\{0,d\},\,a\in\bbR\}.
    \]
    \end{enumerate}
\end{defi}

\vspace{0.05in}

\begin{remark}
We only impose restrictions on the size of the probabilities of the
coefficients of $1$ and $X^d$, which might look surprising at the first
sight. These are the two corners of the support set (Newton polytope) and this assumption turns out to be enough to analyze root isolation algorithms.
We basically set our randomness model this way so that it allows to analyze the most flexible data-model(s). We provide examples below for illustration.
\end{remark}

\begin{exam} \label{ex:uniform}
The uniform random bit polynomial of bitsize $\tau$ we introduced in \Cref{subsec:results1} is the primordial example of a random bit polynomial $\fkf$. For this polynomial we have
$w(\fkf) = \frac{1}{1+2^{\tau+1}}$ and $\tau(\fkf)=\tau$.
\end{exam}

As we will see in the examples below, our randomness model is very flexible.
However, this flexibility comes at a cost. In
principle, we could have $w(\fkf)=1$; this makes our randomness model
equivalent to the worst-case model. To control the effect of large $w(\fkf)$ we
introduce \emph{uniformity}, a quantity to measure how far the leading and trailing coefficient are from the ones of a unifrom random bit polynomial.
\begin{defi}
  \label{def:uniformity}
  The \emph{uniformity} of a random bit polynomial $\fkf$ is
  \[
    u(\fkf):=\ln \left( \left(1+2^{\tau(\fkf)+1} \right) \, w(\fkf)  \right) .
  \]
\end{defi}
\begin{remark}
it holds $u(\fkf)=0$ if and only if the coefficients of $1$
and $X^d$ in $\fkf$ are uniformly distributed in $[-2^{\tau}, 2^{\tau}] \cap
\mathbb{Z}$.
\end{remark}

The following three examples illustrate the flexibility of our random model by
specifying the support, the sign of the coefficients, and their exact bitsize.
Although we specify them separately in the examples, any combination of the specifications is also
possible.

\begin{exam}[Support]\label{ex:spec1}
Let $A\subseteq \{0,1,\ldots,d-1,d\}$ with $0, d\in A$. Then
$\fkf:=\sum_{i\in A}\fkc_iX^i$,
where the $\fkc_i$'s are independent and uniformly distributed in
$[-2^{\tau},2^\tau]$ is a random bit polynomial with $u(\fkf)=0$ and
$\tau(\fkf)=\tau$.
\end{exam}

\begin{exam}[Sign of the coefficients]\label{ex:spec2}
Let $s\in\{-1,+1\}^{d+1}$. The random polynomial $\fkf:=\sum_{i=0}^d\fkc_iX^i$,
where the  $\fkc_i$'s are independent and uniformly distributed in
$s_i([1,2^{\tau}]\cap\bbN)$, is a random bit polynomial with $u(\fkf)\leq
\ln(2)$ 
and $\tau(\fkf)=\tau$.
\end{exam}

\begin{exam}[Exact bitsize]\label{ex:spec3}
Let $\fkf:=\sum_{i=0}^d\fkc_iX^i$ be the random polynomial, where the
$\fkc_i$'s are independent random integers of exact bitsize $\tau$, 
that is, 
$\fkc_i$ is uniformly distributed in
$\bbZ\cap([-2^{\tau}+1,-2^{\tau-1}]\cup[2^{\tau-1},2^{\tau}-1])$. 
Then, $\fkf$
is a random bit polynomial with $u(\fkf)\leq \ln(3)$
 and $\tau(\fkf)=\tau$.
\end{exam}


We consider a \emph{smoothed random model} for polynomials, where a
deterministic polynomial is perturbed by a random one. In this way,
our
random bit polynomial model includes smoothed analysis over integer
coefficients as a special case.

\begin{exam}[Smoothed analysis]\label{ex:smoothed}
Let $f\in \Pd$ be a fixed integer polynomial with coefficients in
$[-2^\tau,2^\tau]$, $\sigma\in\bbZ\setminus \{0\}$ and $\fkf\in\Pd$ a random
bit polynomial. Then,
$
\fkf_\sigma:=f+\sigma\fkf 
$
is a random bit-polynomial with bitsize
$
\tau(\fkf_\sigma)\leq 2 \max\{\tau,\tau(\fkf)+\tau(\sigma)+1 \} ,
$
where $\tau(a)$ denotes the bitsize of $a$, and uniformity
$
u(\fkf_\sigma)\leq 1+u(\fkf) + \max\{\tau-\tau(\fkf),\tau(\sigma)\}.
$
If we combine the smoothed random model with the model of the previous examples, then we can
also consider structured random perturbations.
\end{exam}

Our main results for \descartes, \sturm, \anewdsc, and \jssparse algorithms
are as follows:

\begin{theo}[\descartes solver]
\label{theo:main-Desc}
  Let $\fkf$ be random bit polynomial, of degree $d$, bitsize $\tau(\fkf)$, and
  uniformity parameter $u(\fkf)$, such that
$\tau(\fkf) = \Omega( \log{d}+ u(\fkf))$, then  \descartes solver isolates the
real roots of $\fkf$ in $I = [-1,1]$ in expected time
\[ \sOB(d \, \tau \, (1 + u(\fkf))^3 +d^2\, (1 + u(\fkf))^4). \]
\end{theo}

\begin{remark}
Note that if $\fkf$ is not square-free, \descartes will compute its square-free
part and then proceed as usual to isolate the real roots. The probabilistic complexity estimate covers
this case.
\end{remark}

\begin{theo}[\sturm solver]
  \label{thm:main-Sturm}
  Let $\fkf\in\Pd^\bbZ$ be a random bit polynomial of bit-size $\tau(\fkf) \geq 10$ and
  uniformity $u(\fkf)$.
  If $\tau(\fkf) = \Omega( \log{d}+ u(\fkf))$,
  then the
  expected bit complexity of \sturm to isolate the real roots of $\fkf$ in $I=[-1, 1]$,
  \emph{using fast algorithms for evaluating Sturm sequences}, is
  $\sOB(d^2 \tau(\fkf) \, (1 + u(\fkf))^3)$.\end{theo}

\begin{remark}
	For a "slower" version of \sturm, that is for a variant that does exploits asymptotically fast algorithms for evaluating Sturm sequences, 
	we show a lower bound~\Cref{lem:Sturm-lower-bound}.
	This "slower" version is the one that is commonly implemented. 
\end{remark}

\begin{theo}[\anewdsc solver]
  \label{thm:main-aNewDsc}
Let $\fkf \in \Pd^\bbZ$ be a random bit polynomial with $\tau(\fkf) \geq
\Omega(\log{d} +u(\fkf))$ and
  uniformity $u(\fkf)$. Then, the expected bit complexity of \anewdsc 
  for isolating the real roots of $\fkf$ in $I = [-1,1]$
  is
  $
    \sOB((d^2 + d \, \tau(\fkf)) (1+u(\fkf))^2 ).
  $
\end{theo}

\begin{theo}[\jssparse solver]
	\label{thm:main-jssparse}
Let $\fkf$ be a uniform random bit polynomial of bitsize $\tau$ and $\tau =
\Omega(\log{d} +u(\fkf))$, uniformity $u(\fkf)$,  
having support $|M|$.
Then,  \jssparse computes isolating intervals for all the real roots of $f$ in $I = [-1,1]$ 
    in expected bit complexity 
    $
\sOB\left(|M|^{12} \, \tau^2  \, \log^{3}{d} \right),
$
under the assumption that $\tau > \log^3{d}$.
\end{theo}


\begin{remark}
One might further optimize the probabilistic estimates, that we present in detail in Section~\ref{subsec:prob}, by employing strong tools from Littlewood-Offord
theory~\cite{rudelsonvershynin2008}. However, the complexity analysis depends
on the random variables in a logarithmic scale and so further improvements on
probabilistic estimates will not make any essential improvement on our main
result. Therefore, we prefer to use more transparent proofs with slightly less
optimal dependency on the uniformity parameter $u(\fkf)$.
\end{remark}



\subsection{Overview of main ideas}
There are essentially two  important quantities in analyzing \descartes and the other exact algorithms: the \emph{separation
bound} and the number of complex roots nearby the real axis.

The separation bound is the minimum distance between
the distinct roots of a polynomial~\cite{emt-dmm-j-19}.
This quantity controls the depth of the subdivision tree of \descartes
and we bound it using condition
numbers~\cite{bcssbook,dedieubook,conditionbook,TCTcubeI-journal}. In short, we use condition numbers to obtain an
instance-based estimate for the depth of the subdivision tree of \descartes (and for the other algorithms).
Even though \descartes isolates the real roots, the complex roots
near the real axis control the width of the subdivision tree. 
This follows
from the work of Obreshkoff~\cite{Obreshkoff-book}, see also
\cite{KM-newDesc-06};
for this we call these areas close the real axis Obreshkoff areas.
 To estimate the number of roots in the Obreshkoff areas
we use complex analytic techniques. 
Roughly speaking, by bounding the number of complex
roots in a certain region, we obtain an instance-based
estimate for the width of the subdivision tree of \descartes.
Overall, by controlling both the depth, through the condition number,
and the width, through the number of complex roots in a region around the real axis, we estimate the size of the subsdivision tree of \descartes
which in turn we use to estimate the bit complexity estimate.

Finally, we perform the expected/smoothed analysis of the algorithm \descartes by performing
probabilistic analyses of the number of complex roots and the condition number.
Expected/smoothed analysis results in computational algebraic geometry are rare
and mostly restricted to continuous random variables, with few exceptions
\cite{castromontanapardosanmartin2002};
see also \cite{PanTsi-refine-2013,egt-issac-2010,et-tcs-2007}.
To the best of our knowledge,  we
present the first result for the expected complexity of root finding for
random polynomials with integer coefficients. Our results rely on the strong
toolbox developed by Rudelson, Vershynin, and others in random matrix
theory~\cite{rudelsonvershynin2015,livshytspaourispivovarov2016}.
We use various condition numbers for univariate polynomials
from~\cite{TCTcubeI-journal} to control the separation
bound of random polynomials. However, as mentioned earlierm our probabilistic analysis differs from earlier works, e.g., \cite{conditionbook,TCTcubeI-journal,ergur2021smoothed}, as we consider discrete random perturbations rather than continuous randomness with a density.

Similar arguments as in the case of \descartes apply for the analysis
of the algorithm \textsc{ANewDsc} \cite{sm-anewdsc} (Sec.~\ref{sec:aNewDsc}) that combines Descartes' rule of signs
and Newton operator, as well for the analysis of the sparse solver of Jindal and Sagraloff \cite{JinSag-sparse-17} (Sec.~\ref{sec:jindal-sagraloff}. 
For the \sturm algorithm (Sec.\ref{sec:Sturm}) the important quantities are the number of real roots (as it does not depend on the complex roots at all) and the separation bound. Thus, we also exploit the connection with the condition numbers.

\paragraph*{Organization}
The rest of the paper is structured as follows: In Section 2 we develop our technical toolbox, and in section 3 we perform beyond worst-case analysis of \descartes, \sturm, \anewdsc, and a sparse solver.
\paragraph*{Notation.}
We denote by $\OO$, resp. $\OB$, the arithmetic, resp.  bit,
complexity and we use $\sO$, resp. $\sOB$, to supress (poly-)logarithmic factors. We denote by $\Pd$ the space of univariate polynomials of degree at
most $d$ with real coefficients and by $\Pd^\bbZ$ the subset of integer
polynomial. If $f = \sum_{k=0}^df_k X^k \in \Pd^\bbZ$, then the bitsize of $f$
is the maximum bitsize of its
coefficients. The set of complex roots of $f$ is $\mcZ(f)$,  $f^{(k)}$ the $k$-th derivative of $f$.

We denote by $\var(f)$ the number of sign changes in the
coefficients.
The \emph{separation bound} of $f$, $\Delta(f)$ or $\Delta$ if $f$ is clear from the context, is the minimum distance between the roots of $f$, see \cite{emt-dmm-j-19,ep-dmm-17,Dav:TR:85}. We denote by $\bbD$ the unit disc in the complex plane, by $\bbD(x,r)$ the disk
$x+r\bbD$, and by $I$ the interval $[-1,1]$. For a real interval $J = (a, b)$, we consider $\xmid(J): = \tfrac{a+b}{2}$ and $\wid(J):= b-a$. For a $n \in \NN$. We use $[n]$ for the set $\{1, \cdots, n\}$ and
$\mu(n) = \OB(n\log n)$ for the complexity of multiplying two integers of bitsize $n$.
\section{Condition numbers, separation bounds, and randomness}
\label{sec:cond-sep-prob}
%
%
%

We present a short introduction to condition numbers and we highlight their relation with separation bounds, as well as several deterministic and probabilistic estimates.

First, we introduce the 1-norm for univariate polynomials
and demonstrate how we can use it to bound the coefficients of a Taylor expansion. 
For a polynomial $f\in\Pd$, say
$
f(x) = \sum_{i=0}^d a_i x_{i}
$,
the 1-norm of $f$ is the 1-norm of the vector of its coefficients, 
that is 
$\|f \|_1 = \sum_{i=0}^d |a_i|$.

\begin{prop}\label{prop:1normproperties}
Let $f\in\Pd$ and $x \in I$, then 
\begin{equation}
    \left|\frac{1}{k!}f^{(k)}(x)\right|\leq \binom{d}{k}\|f\|_1.
\end{equation}
\end{prop}
\begin{proof}
	It suffices to observe that $|x| \leq 1$
	and that, for $k \leq \ell \leq d$, 	
	 the $(\ell-k)$-th coefficient of $f^{(k)}$ is the $\ell$-th coefficient of $f$, that is $a_{\ell}$, multiplied by $\ell(\ell -1)\cdots (\ell - k - 1) = \tfrac{\ell!}{(\ell-k)! } \leq \tfrac{d!}{(d-k)!}$.
\end{proof}

\subsection[Condition numbers for univariate polynomials]{Condition numbers for
univariate polynomials}

The \emph{local condition number of $f\in\Pd$ at
  $z\in \bbD$}~\cite{TCTcubeI-journal} is
\begin{equation}
	\label{eq:cond-in-D}
    \cond(f,z):=\frac{\|f\|_1}{\max\{|f(z)|,|f'(z)|/d\}}.
\end{equation}
The same definition using the
$\ell_2$-norm is standard in numerical analysis literature, e.g.,
\cite{Higham-book-02}.

We also define the \emph{(real) global condition number of $f$}  on a domain $I$ as
\begin{equation}
  \label{eq:condR}
  \condR(f):=\max_{x\in I}\cond(f,x).
\end{equation}

We note that as $\condR(f)$ becomes bigger, $f$ is closer to have 
a singular real zero in $I$; we can quantify this using the so-called
condition number theorem, see~\cite[Theorem~4.4]{TCTcubeI-journal}. There are
many interesting properties of $\condR(f)$, but let us state the only one we
will use; we refer to \cite[Theorem~4.2]{TCTcubeI-journal} for additional properties.

\begin{theo}[2nd Lipschitz
property]\cite{TCTcubeI-journal}\label{theo:propconditionnumber}
Let $f\in\Pd$. The map $\bbD\ni z\mapsto 1/\cond(f,z)\in [0,1]$ is well-defined
and $d$-Lipschitz.\eproof
\end{theo}

\subsection{Condition-based estimates for separation}

Next we consider the separation bound of polynomials, e.g., \cite{emt-dmm-j-19}, suitably adjusted in our setting;
it corresponds to the minimum distance between the roots of a polynomial. 
This quantity and its condition-based estimate that follows plays a fundamental role in our complexity estimates.

\begin{defi}\label{defi:realseparation}
For $\varepsilon\in\left[0,\frac{1}{d}\right)$ we set $I_\varepsilon:=\{z\in
\bbC\mid \dist(z,I)\leq \varepsilon\}$. If $f \in \Pd$, then the
\emph{$\varepsilon$-real separation of $f$}, $\Delta\uR_{\varepsilon}(f)$, is
    \[
\Delta_{\varepsilon}\uR(f):=\min\left\{\left|\zeta-\tilde{\zeta}\right|\mid
\zeta,\tilde{\zeta}\in I_\varepsilon ,\,f\left(\zeta\right)
    =f(\tilde{\zeta})=0\right\},
  \]
    if $f$ has no double roots in $I_\varepsilon$, 
    and $\Delta\uR_{\varepsilon}(f):=0$  otherwise.
\end{defi}

\begin{theo}[{\cite[Theorem 6.3]{TCTcubeI-journal}}]
\label{theo:condbasedseparation}
Let $f\in\Pd$ and assume   $\varepsilon\in\left[0, \frac{1}{\enumber
d\condR(f)}\right)$, then $\Delta_{\varepsilon}\uR(f)\geq
\frac{1}{12d\condR(f)}.$\eproof
\end{theo}

\subsection{Probabilistic bounds for condition numbers}\label{subsec:prob}

Next, we introduce our probabilistic framework based on Rudelson and
Vershynin's work ~\cite{rudelsonvershynin2015}.

\begin{theo}\label{theo:realglobalconditionnumberprob_new}
Let $\fkf\in\Pd^{\bbZ}$ be a random bit polynomial. Then, for $t\leq
(d+1)2^{\tau(\fkf)}$,
\[
\bbP(\cond_\bbR(\fkf)\geq t)\leq 8\sqrt{2}\,d(d+1)\enumber^{u(\fkf)}\,\frac{1}{\sqrt{t}}.
\]
\end{theo}

The following corollary gives bounds on all moments of $\log \ln\condR(\fkf)$. It looks somewhat different than Theorem~\ref{theo:realglobalconditionnumberprob_new}, but it has the same essence.


\begin{cor}\label{cor:realglobalexpectations}
Let $\fkf\in\Pd^{\bbZ}$ be a random bit polynomial, $\ell\in\bbN$ and $c\geq
1$. If $\tau(\fkf)\geq 5+3\log(d+1)+3u(\fkf)$, then 
\[
\left(\bbE_\fkf
\left(\min\{\ln\condR(\fkf),c\}\right)^\ell\right)^{\frac{1}{\ell}}
\leq
(\ell+3)
\left(4+3\log(d+1)+2u(\fkf)\right)+\left(\frac{32\,(d+1)^3\enumber^{2u(\fkf)}}{2^{\tau
(\fkf)}}\right)^{\frac{1}{2\ell}}c.
\]
In particular, if $\tau(\fkf)\geq 5+3\log(d+1)+3u(\fkf)+2\ell(\log c-\log \ell)$, then
\[\left(\bbE_\fkf
\left(\min\{\ln\condR(\fkf),c\}\right)^\ell\right)^{\frac{1}{\ell}}\leq
(\ell+3)\left(5+3\log(d+1)+2u(\fkf)\right)=\Oh\left(\ell(\log(d+1)+u(\fkf))\right).\]
\end{cor}
The following comments are in order to understand the limitations of the two theorems and the
corollary above. First, note that
Theorem ~\ref{theo:realglobalconditionnumberprob_new} is meaningful when
\[\tau(\fkf)\geq 2+\frac{3}{2}\log(d) +2u(\fkf) \]
and \Cref{cor:realglobalexpectations} is meaningful when
\[ \tau(\fkf)\geq 5 + 4\log(2)+ 3u(\fkf). \]
Intuitively, the randomness model needs some wiggling room to differ from the
worst-case analysis. In our case this translates roughly to assume that
\[ \tau(\fkf) > \log(d) + u(\fkf). \] 
This is a reasonable assumption because for most cases of
interest, $u(\fkf)$ is bounded above by a constant. In this case, the second condition
in Corollary~\ref{cor:realglobalexpectations} becomes 
\[ \tau(\fkf)=\Omega(\ell\log(d)+\log(c)) .\] 
Moreover, in most application of
Corollary~\ref{cor:realglobalexpectations}, we will have $c=d^{\Oh(1)}$. Hence
we are only imposing roughly  that 
\[ \tau(\fkf) = \Omega( \ln d) .\]

We  need the following proposition for our proofs. Recall that
for $A\in\bbR^{k\times N}$, 
\[ \|A\|_{\infty,\infty}:=\sup_{v\neq
0}\frac{\|Av\|_\infty}{\|v\|_\infty}=\max_{i\in k}\|A^i\|_1 \] 
where $A^i$ is the $i$-th row of $A$.
\begin{prop}\label{prop:reallinearprojectiondiscreterandomvector}
Let $\fkx\in\bbZ^N$ be a random vector with independent coordinates. Assume
that there is a $w>0$ so that for all $i$ and $x\in\bbZ$,
$\bbP(\fkx_i=x)\leq w$.
Then for every linear map $A\in\bbR^{k\times N}$, $b\in \bbR^k$ and
$\varepsilon\in[\|A\|_{\infty,\infty},\infty)$,
\[
\bbP(\|A\fkx+b\|_\infty\leq \varepsilon )\leq
2\frac{(2\sqrt{2}w\varepsilon)^k}{\sqrt{\det AA^*}} .
\]
\end{prop}
\begin{proof}[Proof of
Proposition~\ref{prop:reallinearprojectiondiscreterandomvector}]
Let $\fky\in\bbR^N$ be such that the $\fky_i$ are independent and uniformly
distributed in $(-1/2,1/2)$. Now, a simple computation shows that $\fkx+\fky$
is absolutely continuous and each component has density given by
\[
\delta_{\fkx_i+\fky_i}(t)
=\sum\nolimits_{s\in \bbZ}\bbP(\fkx_i=s)\delta_{\fky_i}(t-s).
\]
Thus each component of $\fkx+\fky$ has density bounded by $w$.
We have
\[
\bbP(\|A\fkx+b\|_\infty\leq \varepsilon)\leq \bbP(\|A(\fkx+\fky)+b\|_\infty\leq
2\varepsilon)/\bbP(\|A\fky\|_\infty\leq \varepsilon),
\]
since $\fkx$ and $\fky$ are independent, and by the triangle inequality.

Now we apply \cite[Proposition~5.2]{TCTcubeI} (which is nothing
more than~\cite[Theorem~1.1]{rudelsonvershynin2015} with the explicit constants
of~\cite{livshytspaourispivovarov2016}): For a random
vector $\fkz\in\bbR^N$ with independent coordinates with density bounded by
$\rho$ and $A\in \bbR^{k\times N}$, we have that $A\fkz$ has density bounded by
$(\sqrt{2}\rho)^k/\sqrt{\det AA^*}$. Thus
\[
\bbP(\|A(\fkx+\fky)+b\|_\infty\leq 2\varepsilon)\leq
(2\sqrt{2}w\varepsilon)^k/\sqrt{\det AA^*}.
\]
On the other hand,
\[
\bbP(\|A\fky\|_\infty\leq \varepsilon)=1-\bbP(\|A\fky\|_\infty\geq
\varepsilon)\geq 1-\bbE\|A\fky\|_\infty/\varepsilon .
\]
by Markov's inequality.
Now, by our assumption on $\varepsilon$, we only need to show that
$\bbE\|A\fky\|_\infty\leq \|A\|_{\infty,\infty}/2$.

By Jensen's inequality,
\[
\bbE\|A\fky\|_\infty=\bbE\lim_{\ell\to\infty}\|A\fky\|_{2\ell}\leq
\lim_{\ell\to\infty}\left(\bbE\|A\fky\|_{2\ell}^{2\ell}\right)^{\frac{1}{2\ell}
} .
\]
Expanding the interior and computing the moments of $\fky$, we obtain
\[
\bbE\|A\fky\|_\infty\leq
\lim_{\ell\to\infty}\left(\sum_{i=1}^k\sum_{|\alpha|=\ell}\binom{2\ell}{2\alpha
}\prod_{j=1}^n\left(A_{i,j}^{2\alpha_j}(1/2)^{2\alpha_j}/(2\alpha_j+1)\right
)\right)^{\frac{1}{2\ell}},
\]
since the odd moments disappear. Thus
\[
\bbE\|A\fky\|_\infty\leq 
\frac{1}{2}\lim_{\ell\to\infty}\left(\sum_{i=1}^k\sum_{|\alpha|=2\ell}\binom
{2\ell}{\alpha}\prod_{j=1}^n\left(|A_{i,j}|^{\alpha_j}\right)\right)^{\frac{1}
{2\ell}}=\frac{\|A\|_{\infty,\infty}}{2},
\]
where we obtained the bound of $\|A\|_{\infty,\infty}/2$ after doing the
binomial sum and taking the limit.
\end{proof}

\begin{proof}[Proof of Theorem~\ref{theo:realglobalconditionnumberprob_new}]
\[
\bbP(\cond(\fkf)\geq t)=\sum_{a_1,\ldots,a_{d-1}}\bbP(\cond(\fkf)\geq t\mid
\fkc_1=a_1,\ldots,\fkc_{d-1}=a_{d-1})\prod_{i=1}^{d-1}\bbP(\fkc_i=a_i).
\]
where $\fkf=\sum_{k=0}^d\fkc_kX^k$. The rest of the proof will deal with a random bit polynomial $\fkf$ of the
form
\[
\fkf=\fkc_0+\sum_{k=1}^{d-1}a_kX^k+\fkc_dX^d
,
\]
where $a_1,\ldots,a_{d-1}\in \bbZ\cap [-2^\tau,2^\tau]$ are arbitrary fixed
integers.

We claim that for a  random $\fkf\in\Pd$ and $t\geq 1$, we have
\begin{equation}\label{eq:EPRtrick}
    \bbP_\fkf(\cond(\fkf)\geq t)\leq 2d\sqrt{t}\,\bbE_{\fkx\in I}\bbP_{\fkf}\left(\frac{|\fkf(\fkx)|}{\|\fkf\|_1}\leq \frac{2}{t}\right).
\end{equation}
We prove this claim as follows: If $\cond(f)\geq t$, then there is $x_\ast\in I$ such that $\cond(f,x_\ast)\geq t$ and then, for $x\in B(x,1/(2d\sqrt{t}))\cap I$,
\begin{align*}
    \frac{|f(x)|}{\|f\|_1}&\leq \frac{|f(x_\ast)|}{\|f\|_1}+\frac{|f'(x_\ast)|}{\|f\|_1}|x-x_\ast|+\frac{1}{2}\max_{\xi\in I}\frac{|f''(\xi)|}{\|f\|_1}|x-x_\ast|^2&\text{(Taylor's theorem)}\\
    &\leq \frac{|f(x_\ast)|}{\|f\|_1}+\frac{|f'(x_\ast)|}{\|f\|_1}\frac{1}{2d\sqrt{t}}+\frac{1}{2}\max_{\xi\in I}\frac{|f''(\xi)|}{\|f\|_1}\frac{1}{4d^2t}&(|x-x_\ast|\leq 1/(2d\sqrt{t}))\\
    &\leq \frac{1}{\cond(f,x_\ast)}\left(1+\frac{1}{2\sqrt{t}}\right)+\frac{1}{2}\max_{\xi\in I}\frac{|f''(\xi)|}{\|f\|_1}\frac{1}{4d^2t}\\
    &\leq \frac{1}{t}\left(1+\frac{1}{2\sqrt{t}}\right)+\frac{1}{2}\max_{\xi\in I}\frac{|f''(\xi)|}{\|f\|_1}\frac{1}{4d^2t}\\
    &\leq \frac{1}{t}\left(1+\frac{1}{2\sqrt{t}}\right)+\frac{1}{8t}&\text{(Proposition~\ref{prop:1normproperties})}\\
    &=\frac{1}{t}\left(1+\frac{1}{8}+\frac{1}{2\sqrt{t}}\right) \leq \frac{2}{t}.
\end{align*}
Hence $\cond(f)\geq t$ implies $\bbP_{\fkx\in I}\left(\frac{|\fkf(\fkx)|}{\|\fkf\|_1}\leq \frac{2}{t}\right)\geq 1/(2d\sqrt{t})$, and thus
\begin{align*}
    \bbP_\fkf(\cond(\fkf)\geq t)&\leq \bbP_{\fkf}\left(\bbP_{\fkx\in I}\left(\frac{|\fkf(\fkx)|}{\|\fkf\|_1}\leq \frac{2}{t}\right)\geq \frac{1}{2d\sqrt{t}}\right)&\text{(Implication bound)}\\
    &\leq 2d\sqrt{t} \, \bbE_{\fkf}\bbP_{\fkx\in I}\left(\frac{|\fkf(\fkx)|}{\|\fkf\|_1}\leq \frac{2}{t}\right)&\text{(Markov's inequality)}\\
    &\leq 2d\sqrt{t}\, \bbE_{\fkx\in I}\bbP_{\fkf}\left(\frac{|\fkf(\fkx)|}{\|\fkf\|_1}\leq \frac{2}{t}\right),&\text{(Tonelli's theorem)}
\end{align*}
Now, let $\Pd(a_1,\ldots,a_{d-1})$ be the affine subspace of $\Pd$ given by the equations
$f_k=a_k$ for $k\in\{1,\ldots,d-1\}$, and let 
$
f\mapsto Af+b
$
be the affine mapping given by
\[
\Pd(a_1,\ldots,a_{d-1})\ni f\mapsto f(x)\in\bbR.
\]
In the coordinates we are working on (those of the
base $\{1,X^d\}$), $A$ has the form
$
\begin{pmatrix}
1&x^d
\end{pmatrix},
$
and so, by an elementary computation, we have $\|A\|_{\infty,\infty}=1+|x|^d\leq 2$ and $\sqrt{\det AA^*}= 1+|x|^{2d}\geq 1$. Now, since $\|\fkf\|_1\leq (d+1)2^{\tau(\fkf)}$, we have that
\begin{equation}
    \bbP_{\fkf}\left(\frac{|\fkf(x)|}{\|\fkf\|_1}\leq \frac{2}{t}\right)=\bbP_{\fkf}\left(\|A\fkf+b\|_\infty\leq \frac{2}{t}\|\fkf\|_1\right)\leq \bbP_{\fkf}\left(\|A\fkf+b\|_\infty\leq (d+1)2^{\tau(\fkf)+1}\frac{1}{t}\right),
\end{equation}
and so, by~\eqref{eq:EPRtrick} above,
\[
\bbP_\fkf(\cond(\fkf)\geq t)\leq 2d\sqrt{t}\,\bbE_{\fkx\in I}\bbP_{\fkf}\left(\|A\fkf+b\|_\infty\leq (d+1)2^{\tau(\fkf)+1}\frac{1}{t}\right).
\]
Therefore, by Proposition~\ref{prop:reallinearprojectiondiscreterandomvector}, we have that for $t\leq (d+1)2^{\tau(\fkf)}$,
\[
\bbP_\fkf(\cond(\fkf)\geq t)\leq 8\sqrt{2}d(d+1)\enumber^{u(\fkf)}\frac{1}{\sqrt{t}},
\]
where we have applied the definition of $u(\fkf)$. Hence the desired result follows.
\end{proof}
\begin{proof}[Proof of Corollary~\ref{cor:realglobalexpectations}]
For $\fkx=\log\condR(\fkf)$, 
\[
U:=2\ln(8\sqrt{2}\,d(d+1)\enumber^{u(\fkf)})\leq 4\ln(ed)+u(\fkf)
\]
and $V:=\ln((d+1)2^{\tau(\fkf)})$ using the  assumption $u(\fkf)\geq 0$ and Theorem~\ref{theo:realglobalconditionnumberprob_new} we that for any $s\in [U,V]$ 
\[
\bbP(\fkx\geq s)\leq \enumber^{\frac{U-s}{2}}.
\]
So, to complete the proof it is enough to show the following: Let $2\leq U\leq V$ and $c\geq 1$ and $\fkx$ be a positive random variable such that for $s\in [U,V]$, 
\[ \bbP(\fkx\geq s)\leq \enumber^{\frac{U-s}{2}} \Rightarrow  \bbE\min\{\fkx,c\}^\ell\leq U^\ell+\enumber \ell! U^{\ell}+e^{\frac{U-V}{2}}c^\ell.\]

Since the value of the expectation grows with $c$, we can assume, without loss of generality, that $V<c.$ Otherwise, the value would be smaller and the same bound would be valid. 
\begin{multline}
\bbE_\fkf \left(\min\{\fkx,c\}\right)^\ell=\int_{0}^\infty \ell
s^{\ell-1}\bbP(\min\{\fkx,c\}\geq s)\,\mathrm{d}s
=\int_{0}^c \ell
s^{\ell-1}\bbP(\fkx\geq s)\,\mathrm{d}s\\
=\int_{0}^U \ell
s^{\ell-1}\bbP(\fkx\geq s)\,\mathrm{d}s
+\int_{U}^V \ell
s^{\ell-1}\bbP(\fkx\geq s)\,\mathrm{d}s
+\int_{V}^c \ell
s^{\ell-1}\bbP(\fkx\geq s)\,\mathrm{d}s.  
\end{multline}
where the first equality follows from the fact that $\fkx$ is a positive random variable, and the second one from the fact that for $s\geq c$, $\bbP(\min\{\fkx,c\}\geq s)=0$; and for $s\leq c$, $\bbP(\min\{\fkx,c\}\geq s)=\bbP(\fkx\geq s)$.

In $[0,U]$, we have that
\[
\int_{0}^U \ell
s^{\ell-1}\bbP(\fkx\geq s)\,\mathrm{d}s\leq \int_{0}^U \ell
s^{\ell-1}\,\mathrm{d}s\leq U^\ell,
\]
since the probability is always bounded by 1. In $[U,V]$, we have that
\begingroup
\allowdisplaybreaks
\begin{align*}
\int_{U}^V \ell
s^{\ell-1}\bbP(\fkx\geq s)\,\mathrm{d}s
&\leq \int_{U}^V \ell s^{\ell-1}e^{\frac{U-s}{2}}\,\mathrm{d}s&\text{(Assumption on }\fkx\text{)}\\
&=\int_{0}^{V-U} \ell (s+U)^{\ell-1}e^{-s/2}\,\mathrm{d}s&\text{(Change of
variables)}\\
&\leq \int_{0}^{\infty} \ell
(s+U)^{\ell-1}e^{-s/2}\,\mathrm{d}s&\text{(Non-negative integrand)}\\
&\leq \ell \sum_{k=0}^{\ell-1}\binom{\ell-1}{k}U^{\ell-1-k}\int_0^\infty s^k
e^{-s/2}\,\mathrm{d}s&\text{(binomial identity)}\\
&=\ell \sum_{k=0}^{\ell-1}\binom{\ell-1}{k}k!U^{\ell-1-k}2^{k+1}&\text{(Euler's
Gamma)}\\
&\leq\ell \sum_{k=0}^{\ell-1}\binom{\ell-1}{k}k!U^{\ell}&(2\leq U)\\
&\leq \ell! U^{\ell-1}\sum_{k=0}^{\ell-1}\frac{1}{(\ell-1-k)!}&\left(\binom{\ell}{k}k!=\frac{(\ell-1)!}{(\ell-1-k)!}
\ell^{\ell-1},\,U>1\right)\\
&=\ell! U^{\ell}\sum_{k=0}^{\ell-1}\frac{1}{k!} \leq \enumber \ell! U^{\ell}
\enspace .
\end{align*} 
\endgroup
Hence
\[
\int_{U}^V \ell
s^{\ell-1}\bbP(\fkx\geq s)\,\mathrm{d}s\leq \enumber \ell! U^{\ell}.
\]

In $[V,c]$, we have that
\[
\int_{V}^c \ell
s^{\ell-1}\bbP(\fkx\geq s)\,\mathrm{d}s
\leq \int_{V}^c \ell
s^{\ell-1}\bbP(\fkx\geq V)\,\mathrm{d}s
\leq \int_{V}^c \ell
s^{\ell-1}\enumber^{\frac{U-V}{2}}\,\mathrm{d}s=\enumber^{\frac{U-V}{2}}\left(c^\ell-V^\ell\right)\leq \enumber^{\frac{U-V}{2}}c^\ell.
\]
Therefore, since $e^{U-V}\int_{V}^c \ell s^{\ell-1}\,\mathrm{d}s\leq
e^{U-V}\int_{0}^c \ell s^{\ell-1}\,\mathrm{d}s$,
\[
\int_{V}^c \ell s^{\ell-1}\bbP(\min\{\ln\condR(\fkf),c\}\geq
s)\,\mathrm{d}s\leq e^{U-V}c^{\ell} .
\]
%
To obtain the final estimate, we add the three upper bounds obtaining the uper
bound
$
U^\ell+\ell^\ell U^{\ell-1}+e^{U-V}c^\ell.
$
After substituting the values of $U$ and $V$ and some easy estimations, we
conclude.
\end{proof}

\subsection{Bounds on the number of complex roots close to real axis}\label{sec:complexroots}
We need to control the number of roots that are close to real axis to be able analyze \descartes. We use tools from complex analysis together with tools developed in this paper on probabilistic analysis condition numbers. Note that we cannot bound the number of complex roots inside complex disk of constant radius; the symmetry on our randomness model forces any bound to be
of the form $\Oh(d)$. So, inspired by
\cite{moroz2021}, we consider a family of "hyperbolic" disks
$\{\bbD(\xi_{n,N},\rho_{n,N})\}_{n=-N}^N$; we will specify  $N \in \NN$ in the sequel.
In particular,
\begin{equation}
\label{eq:Disc-center}
\xi_{n,N}=\begin{cases}
\sgn(n)\left(1-\frac{3}{4}\frac{1}{2^{|n|}}\right),&\text{if }|n|\leq N-1\\
\sgn(n)\left(1-\frac{1}{2^{N}}\right),&\text{if }|n| = N
\end{cases}
\end{equation}
\begin{equation}
\label{eq:Disc-radius}
\rho_{n,N}=\begin{cases}
\frac{3}{8}\frac{1}{2^{|n|}},&\text{if }|n|\leq N-1\\
\frac{3}{2}\frac{1}{2^{N}},&\text{if }|n| = N
\end{cases}.
\end{equation}
We will abuse notation and write $\xi_n$ and $\rho_n$ instead of $\xi_{n,N}$
and $\rho_{n,N}$ since we will not be working with different $N$'s at the same
time, but only with one $N$ which might not have a prefixed value. For this
family of disks, we will give a deterministic and a probabilistic bound for the
number of roots, $\varrho(f)$, in their union, when $N = \lceil \log{d} \rceil$;
in particular
\begin{equation}
\label{eq:Omega_N}
\varrho(f):=\#\Bigg\{ z\in\Omega_d:=\bigcup_{n=-\lceil\log d\rceil}^{\lceil\log
d\rceil}\bbD(\xi_{n},\rho_{n})\mid f(z)=0 \Bigg\} ,
\end{equation}
where $f\in \Pd$.
We use these bounds to estimate the number of steps of \nameref{alg:Descartes}.

\subsubsection{Deterministic bound}

\begin{theo}\label{theo:deterministicbopundroots}
Let $f\in\Pd$. Then
\[
\varrho(f)\leq \sum_{n=-\lceil\log d\rceil}^{\lceil\log d\rceil}\log
\frac{\enumber\|f\|_1}{|f(\xi_{n})|}.
\]
\end{theo}

We need the following lemma.
\begin{lem}\label{lem:rootsinsmalldisk}
Let $f\in\Pd$, $\xi\in \bbD$, and $\rho>0$. 
If $|\xi|+2\rho<1+1/d$, then
\[ 	\#(\mcZ(f)\cap \bbD(\xi,\rho))
	\leq \log \left(  \frac{ \enumber\|f\|_1}{|f(\xi)|} \right)
  \]
\end{lem}

\begin{proof}[Proof of Lemma~\ref{lem:rootsinsmalldisk}]
We use a classic result of Titchmarsh~\cite[p.~171]{titchmarsh1939} that bounds
the number of roots in a disk. For $\delta\in(0,1)$, we have that
\[ \#(\mcZ(f)\cap \bbD(\xi,\rho))\leq (\ln(1/\delta))^{-1}\ln
(\max_{z\in\bbD}|f(\xi+\rho z/\delta)|/|f(\xi)|).
 \]
where $\bbD$ denotes the unit disk.

We take $\delta=1/2$, and by our assumption on $\xi,\rho$ we have $\xi+2\rho \bbD\in (1+1/d)\bbD$. Since $|f(z)|\leq \enumber \|f\|_1$, for $z\in(1+1/d)\bbD$~\cite[Proposition
3.9.]{TCTcubeI-journal} this gives the following:
\[ \max_{z\in\bbD}|f(\xi+\rho z/\delta)|\leq \max_{z\in(1+1/d)\bbD}|f(z)|\leq
\enumber \|f\|_1 .
 \]
\end{proof}
\begin{proof}[Proof of Theorem~\ref{theo:deterministicbopundroots}]
We only have to apply subadditivity and Lemma~\ref{lem:rootsinsmalldisk}. Note
that the condition of the Lemma~\ref{lem:rootsinsmalldisk} holds for every disk
$\bbD(\xi_{n},\rho_{n})$ in $\Omega_d$.
\end{proof}

\subsubsection{Probabilistic bound}

\begin{theo}\label{theo:probbopundroots}
Let $\fkf\in\Pd^{\bbZ}$ be a random bit polynomial. Then for all $t\leq
\tau(\fkf)(2\lceil\log d\rceil+1)$,
\[
\bbP\left(\varrho(\fkf)\geq t\right)
\leq 
44d^2{(2\lceil\log d\rceil+1)}\enumber^{u(\fkf)}\enumber^{-\frac{t}{{2\lceil\log
d\rceil+1}}}.
\]
\end{theo}
\begin{cor}\label{cor:probboundroots}
Let $\fkf\in\Pd^{\bbZ}$ be a random bit polynomial and $\ell\in \bbN$. Suppose
that $\tau(\fkf)\geq 10\ln(\enumber d)+2u(\fkf)$. Then
\[
\Big(\bbE\varrho(\fkf)^\ell\Big)^{\frac{1}{\ell}}
\leq 
2(1+\ell)(6\ln(ed)+u(\fkf))\ln(\enumber
d)+\Big(\frac{44d^{3+2\ell}\enumber^{u(\fkf)}}{2^{\tau(\fkf)}}\Big)^{\frac{1}
{\ell}}.
\]
In particular, if $\tau(\fkf)\geq (9+3\ell)\ln(\enumber d)+2u(\fkf)$, then
\[
\left(\bbE\varrho(\fkf)^\ell\right)^{\frac{1}{\ell}}\leq \Oh\left(\ell (\ln
d+u(\fkf))\ln d\right).
\]
\end{cor}

\begin{proof}[Proof of Theorem~\ref{theo:probbopundroots}]
If $\#\left(\mcZ(\fkf)\cap \Omega_d\right)\geq t$, then, by
Theorem~\ref{theo:deterministicbopundroots}, there is an $n$ such that
$
\log (\enumber\|f\|_1/|\fkf(\xi_{n})|)\geq t/(2\lceil\log d\rceil+1)
$.
Hence
\[
\bbP\left(\varrho(\fkf)\geq t\right)\leq 
\sum_{n=-\lceil\log d\rceil}^{\lceil\log d\rceil} \bbP\left(\log
\frac{\enumber\|f\|_1}{|\fkf(\xi_{n})|}
\geq \frac{t}{2\lceil\log d\rceil+1}\right).
\]
Now, fix $x\in I$. We argue as in the proof of
Theorem~\ref{theo:realglobalconditionnumberprob_new}, but we consider that map
mapping $f$ to $f(x)$, so that our matrix $A$ takes the form
\[\begin{pmatrix}
1&x&x^{d-1}&x^d
\end{pmatrix}.\]
Note that this $A$ has $\|A\|_{\infty,\infty}\leq d+1$. So, we can apply
Proposition~\ref{prop:reallinearprojectiondiscreterandomvector} to show that
for any $s\leq 2^{\tau(\fkf)}$,
                \[
\bbP\left(\enumber\|\fkf\|_1/|\fkf(x)|\geq s\right)\leq
44d^2\enumber^{u(\fkf)}/s.
\]
If $s=\enumber^{t/N}$, 
with $N = 2\lceil\log(d) \rceil +1$, then the bound follows.
\end{proof}
\begin{proof}[Proof of Corollary~\ref{cor:probboundroots}]
In the proof of Corollary~\ref{cor:realglobalexpectations}  we only used the
fact that the tail bound is of the form $U\enumber^{-t}$ for $t\leq V$ with
$U\leq V$. We will use a similar idea in this proof. Let $0\leq U\leq V$,
$c>0$, and $\fkx\in[0,\infty)$ a random variable. If $\bbP(\fkx\geq t)\leq
\enumber^{U-s}$ for $s\leq V$, then
$\bbE(\min\{\fkx,c\})^\ell\leq U^\ell+\ell^\ell
U^{\ell-1}+\enumber^{U-V}c^\ell$.

By Theorem~\ref{theo:probbopundroots}, the random variable
$\varrho(\fkf)/(2\lceil\log d\rceil+1)$
satisfies the conditions to be a random variable $\fkx$ with
$U=\ln(44d^2(2\lceil\log d\rceil+1)\enumber^{u(\fkf)})\leq 4\ln(\enumber d)+\ln
(2\lceil\log d\rceil+1)+u(\fkf)$, $V=\ln(2^{\tau(\fkf)}/(2\lceil\log
d\rceil+1))$, and $c=\frac{d}{(2\lceil\log d\rceil+1)}$; since the roots are at
most $d$. By our assumptions $U\leq V$,
that~concludes~the~proof.
\end{proof}

\section{Beyond Worst-Case Analysis of Root Isolation Algorithms}
The main idea behind the subdivision algorithms for real root isolation is the
binary search algorithm. We consider an oracle that can guess the number of real
roots in an interval (it can even overestimate them). We keep subdividing the
initial interval until the estimated, by the oracle, number of real roots is
either 0 or 1. Different realizations of the oracle lead to different

In what follows, consider \descartes solver (Section~\ref{sec:descartes}),
the \sturm solver (Section~\ref{sec:Sturm}), 
\anewdsc solver (Section~\ref{sec:aNewDsc}), 
and solver for sparse polynomials by Jindal and Sagraloff (Section~\ref{sec:jindal-sagraloff}).

\subsection{The \textsc{Descartes} solver}
\label{sec:descartes}

The \descartes solver is an algorithm that is based on Descartes' rule of signs.

\begin{theo}[Descartes' rule of signs]
  \label{thm:Desc-rule-of-sign}
  The number of sign variations in the coefficients' list of a
  polynomial $f = \sum_{i=0}^d f_i \, X^i \in \Pd$ equals the number of
  positive real roots (counting multiplicities) of $f$, say $r$,
  plus an even number; that is
  $r \equiv \var(f) \mod 2$.\eproof
\end{theo}

\begin{algorithm2e}[t]
  \scriptsize \dontprintsemicolon \linesnumbered
  \SetFuncSty{textsc} \SetKw{RET}{{\sc return}} \SetKw{OUT}{{\sc output \ }}
  \SetVline \KwIn{A square-free polynomial $f\in \Pd^\bbZ$}
\KwOut{A list, $S$, of isolating intervals for the real roots of $f$ in $J_0 =
(-1, 1)$}

  \BlankLine

  $J_0 \leftarrow (-1, 1)$, 
  $S \leftarrow \emptyset,\, Q \leftarrow \emptyset$,
  $Q \leftarrow \FuncSty{push}( {J_0})$ \;

  \While{ $Q \neq \emptyset$}{
    \nllabel{alg:Subdivision-while-loop}

    ${J}=(a,b) \leftarrow \FuncSty{pop}( Q)$ \\
    $V \leftarrow \var(f, J)$ \;

    \Switch{ $V$ }{

      \lCase{ $V = 0$ }{ \KwSty{continue}\; }
      \lCase{ $V = 1$ }{ $S \leftarrow \FuncSty{ add}( {I})$ \; }
      \Case{ $V > 1$ } {
        $m \gets \frac{a+b}{2}$ \;
        \lIf{$f(m) = 0$}{  $S \leftarrow \FuncSty{ add}( {[m, m]})$ \; }
        $J_L \gets [a, m]$ ;   $J_R \gets [m, b]$ \;

        $Q \leftarrow \FuncSty{push}( Q, {J_L} )$,
        $Q \leftarrow \FuncSty{push}( Q, {J_R} )$ \;
      }
    }
  }
  \RET $S$ \;
  \caption{ $\func{Descartes}(f)$}
  \label{alg:Descartes}
\end{algorithm2e}

In general, Theorem~\ref{thm:Desc-rule-of-sign} provides an
overestimation on the number of positive real roots. It counts exactly
when the number of sign variations is 0 or 1
and if the polynomial is hyperbolic, that is it has \emph{only} real roots.
To count the real roots of $f$ in an interval
$J = (a, b)$,
we use the transformation $x \mapsto \frac{ax + b}{x + 1}$
that maps $J$ to $(0, \infty)$.
Then,
\[\var(f, J) := \var( (X + 1)^d f(\tfrac{aX + b}{X + 1}) )\]
bounds the number of real roots of $f$ in $J$.

Therefore, to isolate the real roots of $f$ in an interval, say  $J_0 = (-1,
1)$,
we count (actually bound) the number of roots of $f$ in $J_0$ using $V =
\var(f, J_0)$.
If $V = 0$, then we discard the interval.
If $V = 1$, then we add $J_0$ to the list of isolating intervals.
If $V > 1$, then we subdivide the interval to two intervals $J_L$ and $J_R$
and we repeat the process.
If we the middle of an interval is a root, then we can detect this by evaluation. 
Notice that in this case we have found a rational root.
The pseudo-code of \descartes appears in Algorithm~\ref{alg:Descartes}.

The recursive process of the \descartes defines a binary tree. Every node of the
tree corresponds to an interval. The root corresponds to the initial interval
$J_0=(-1, 1)$. If a node corresponds to an interval $J=(a,b)$, then its
children correspond to the open left and right half intervals of $J$, that is
$J_L = (a, \xmid(J))$ and $J_R = (\xmid(J), b)$ respectively.
The internal nodes of the tree correspond to intervals $J$, such that
$\var(f,J) \geq 2$.
The leafs correspond to intervals that contain 0 or 1 real roots of $f$. 
Overall, the number of nodes of the tree correspond to the number of
steps, i.e., subdivisions, that the algorithm performs.  We control the number
of nodes by controlling
the depth of tree and the width of every layer. Hence, to obtain the
final complexity estimate it suffices to multiply the number of steps (width
times height) with the
worst case cost of each step.

The following proposition helps to control the cost of each step. Note that at
each step we perform a Mobius transformation
and we do the sign counting at the resulting polynomial.

\begin{prop}\label{prop:polymaps}
  \label{prop:poly-maps}
  Let $f=\sum_{i=0}^df_iX^i\in\Pd^\bbZ$ of bit-size $\tau$.
  \begin{itemize}[leftmargin=*]
    \item The reciprocal transformation is
      $R(f): = X^{d} f(\tfrac{1}{X}) = \sum_{k=0}^{d}f_{d-k}X^k$. Its cost is
      $\OB(1)$ and it does not alter neither the degree nor the bit-size of the
      polynomial.
    \item The homothetic transformation of $f$ by $2^k$, for a positive integer
      $k$, is
$H_k(f) = 2^{d k} f(\tfrac{X}{2^k}) = \sum_{i=0}^d{2^{k(d-i)} f_i \, X^{i}}$.
      It costs $\OB(d \, \mu(\tau + d k)) = \sOB(d\tau + d^2 k)$ and the
      resulting polynomial has bit-size $\OO(\tau + d k)$. Notice that
      $H_{-k} = R H_{k} R$.

\item The Taylor shift of $f$ by in integer $c$ is $T_c(f) = f(x+c) =
\sum_{k=0}^d{a_k x^k}$,
where $a_i = \sum_{j=i}^{d}\binom{j}{i}f_j c^{j-i}$ for $0 \leq i \leq d$. It
costs $\OB(\mu(d^2\sigma + d\tau)\log{d}) = \sOB(d^2\sigma + d\tau)$
     \cite[Corollary~2.5]{vzGGer}, where $\sigma$ is the bit-size of $c$. The
     resulting polynomial has bit-size $\OO(\tau + d\sigma)$.\eproof
\item Given a polynomial $f(x)$ of degree $d$ and bit-size $\tau$, the bit complexity of evaluating a $f$ at a rational point of bit-size $\sigma$ is $\mathrm{O}(d(\tau+\sigma))$ \cite{bz-upol-eval-11,hn-upol-eval-11}.
  \end{itemize}
\end{prop}

\begin{remark}
  There is no restriction on working with open intervals since we
consider an integer polynomial and we can always evaluate it at the endpoints.
  Moreover, to isolate all the real roots of $f$ it suffices to have a routine to
  isolate the real roots in $(-1, 1)$; using the map $x \mapsto 1/x$ we can
  isolate the roots in $(-\infty, -1)$ and  $(1, \infty)$.
\end{remark}

\subsubsection{Bounds on the number of sign variations}
\label{sec:Desc-termination}

For this subsection we consider $f=\sum_{i=0}^df_iX^i\in\Pd$
to be a polynomial with real coefficients, not necessarily integers.
To establish the termination and estimate the bit complexity of \descartes
we need to introduce the Obreshkoff area and lens. Our presentation follows closely \cite{sm-anewdsc,KM-newDesc-06,emt-lncs-2006}.

\begin{figure}[h]
  \centering
  \includegraphics[scale=0.4]{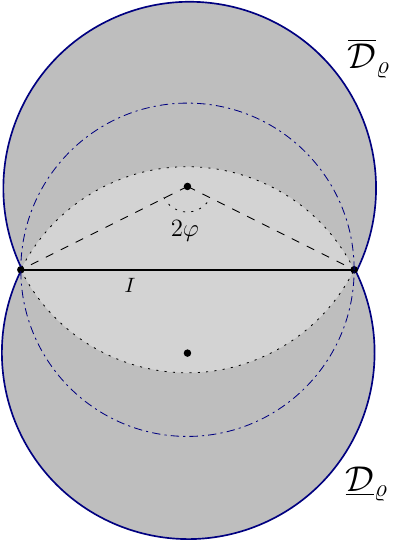}
\caption{\rm{Obreshkoff discs, lens (light grey), and area  (light grey and grey)
for an interval $I$.}}
  \label{fig:Obreshkoff}
\end{figure}

Consider $0 \leq \alpha \leq d$ and a real open interval $J = (a, b)$. The
\emph{Obreshkoff discs}, $\uObD_{\alpha} $ and $\dObD_{\alpha}$, are discs with
boundaries going through the endpoints of $J$. Their centers are above,
respectively below, $J$ and they form an angle $\varphi = \frac{\pi}{\alpha+2}$
with the endpoints of $I$. Its diameter is
$\wid(J)/ \sin(\frac{\pi}{\alpha+2})$.

The \emph{Obreshkoff area} is $\ObA_{\varrho}(J) =
\mathsf{interior}(\uObD_{\alpha} \,\cup\, \dObD_{\alpha})$;
it appears with grey color in Fig.~\ref{fig:Obreshkoff}.
The \emph{Obreshkoff lens} is $\ObL_{\alpha}(J) =
\mathsf{interior}(\uObD_{\alpha} \,\cap\, \dObD_{\alpha})$;
it appears in light-grey color in Fig.~\ref{fig:Obreshkoff}.
If it is clear from the context, then we omit $I$ and we write $\ObA_{\alpha}$
and  $\ObL_{\alpha}$,
instead of $\ObA_{\alpha}(J)$ and $\ObL_{\alpha}(J)$.
It holds that
$\ObL_d \subset \ObL_{d-1} \subset \cdots \subset \ObL_1 \subset \ObL_0$
and
$\ObA_0 \subset \ObA_1 \subset \cdots \subset \ObA_{d-1} \subset \ObA_d$.

The following theorem shows the role of complex roots in the
control of the number of variation signs.

\begin{theo}[\cite{Obreshkoff-book}]
  \label{thm:Obr-signs}
  Consider $f\in\Pd$ 
  and real open interval $J=(a,b)$.
If the Obreshkoff lens $\ObL_{d - k}$ contains at least $k$ roots (counted with
multiplicity) of $f$,
  then $k \leq \var(f,J)$.
If the Obreshkoff area $\ObA_{k}$ contains at most $k$ roots (counted with
multiplicity) of $f$,
  then $\var(f,J) \leq k$.
  Especially
  \begin{equation*}\tag*{\qed}
    \# \{ \text{roots of }f\text{ in } \ObL_d \}
    \le  \var(f,J) \le
    \# \{ \text{roots of }f\text{ in } \ObA_d \}.
  \end{equation*}
\end{theo}

This theorem together with the subadditive property
of Descartes' rule of signs (Thm.~\ref{thm:desc-subadd})
shows that the number of complex roots in the Obreshkoff areas
controls the width of the subdivision tree of \descartes.

\begin{theo}
  \label{thm:desc-subadd}
  Consider a real polynomial $f\in\Pd$.
Let $J$ be a real interval and $J_1, \dots, J_n$ be disjoint open subintervals
of $J$.
  Then, it holds $\sum_{i=1}^n \var(f, J_i) \leq \var(f,J)$.\eproof
\end{theo}

Finally, to control the depth of the subdivision tree of \descartes we use the
one and two circle theorem~\cite{AleGalu-vincent-98,KM-newDesc-06}.
We present a variant
based on the $\varepsilon$-real separation of $f$, $\Delta_{\varepsilon}\uR(f)$
(Definition~\ref{defi:realseparation}).

\begin{theo}\label{thm:sepdepthbound}
Consider $f\in\Pd$, an interval $J\subseteq (-1,1)$ and $\varepsilon>0$. If
\[
2 \, \wid(J)\leq \min\{\Delta_{\varepsilon}\uR(f),\varepsilon\} ,
\]
then either $\var(f, J)=0$ (and $J$ does not contain any real root), or $\var(f, J)=1$
(and $J$ contains exactly one real root).
\end{theo}
\begin{proof}
The proof follows the same application of the one and two circle theorems as in
the proof of \cite[Proposition~6.4]{TCTcubeI-journal}.
\end{proof}

\subsubsection{Complexity estimates for \descartes}
\label{sec:Descartes-complexity-exp}
We give a high-level overview of the proof ideas of this section before going
into technical details. The process of \descartes corresponds to a binary tree
and we control its depth using the real condition number through ~\Cref{theo:condbasedseparation} and~\Cref{thm:sepdepthbound}.
To bound the width of the \descartes' tree we use the Obreskoff areas and the
number of complex roots in them (Theorem~\ref{thm:Obr-signs}).
By combining these two bounds, we control the size of the tree and so we obtain
an instance-based complexity estimate. To turn this instance-based complexity
estimate into an expected (or smoothed) analysis estimation, we use
\Cref{theo:realglobalconditionnumberprob_new}, \Cref{theo:probbopundroots}, \Cref{cor:realglobalexpectations}, and \Cref{cor:probboundroots}.
\paragraph{Instance-based estimates}
\begin{theo}
 \label{thm:Descartes-steps}
If $f\in\Pd^\bbZ$, then,  using  
\descartes,  the number of subdivision steps to isolate the real roots in $I =
(-1, 1)$
	is 	 
	\[
	\sO(\varrho(f)^2\log(\condR(f)).
	\] 
The bit complexity of the algorithm is
  \[
    \sOB(d \tau \varrho(f)^2 \log \condR(f)+ d^2 \varrho(f)^2 \log^2\condR(f)).
  \]
\end{theo}

The definition of the real global condition number, $\condR(f)$,
appears in~\eqref{eq:condR} and 
the definition of the number of roots of $f$ in a family of hyperbolic discs, 
$\varrho(f)$, appears in~\eqref{eq:Omega_N}.

\begin{figure}[h]
  \centering
  \begin{subfigure}[b]{0.5\textwidth}
    \centering
    \includegraphics[scale=0.65]{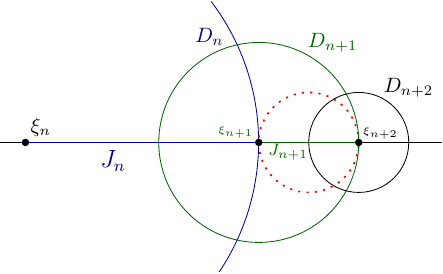}
    
        \label{fig:covering-discs}
  \end{subfigure}%
  \hfill
  \begin{subfigure}[b]{0.5\textwidth}
    \centering
    \includegraphics[scale=0.65]{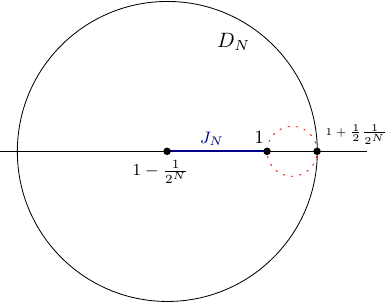}
    \label{fig:covering-last-disc}
  \end{subfigure}
  \label{fig:covering-discs-all}
\caption{\rm{Covering discs of the interval $I = (0, 1)$. (left) Three covering
discs, $D_n, D_{n+1}, D_{n+2}$. 
(right) The (red) dotted circle is the
auxiliary disc that we  ensure  is contained in $D_{n+1} \setminus D_{n}$.}}
\end{figure}

\begin{proof}
We consider the number of steps to isolate the real roots in $I = (-1, 1)$. Let
$N=\ceil{\log d}$ and $\varrho=\varrho(f)$ the number of complex roots in
$\Omega_d$. Recall that $\Omega_d$ is the union of the discs
$D_n:=\bbD(\xi_n,\rho_n):=\xi_n+\rho_n\bbD$, where $\abs{n} \leq N$; see
\eqref{eq:Disc-center} and \eqref{eq:Disc-radius} for the concrete formulas,
and that it contains the interval $I$.

The discs partition $I$ into the $2N+1$ subintervals $J_n := [\xi_n,
\xi_{n+1}]$ (or $J_n := [\xi_n, \xi_{n-1}]$ if $n \leq 0$). Note that $J_n$ is
the union of $3$ intervals of size $1/2^{n+3}$. Because of this, there is a
binary subdivision tree of $I$ of size $\Oh(\log^2 d)$ such that every of its
intervals is contained in some $J_n$. Thus, if we bound the width of the
subdivision tree of \descartes starting at each $J_n$ by $w$, then the width of
the subdivision tree of \descartes starting at $I$ is bounded by $\Oh(w\log^2
d+\log^2 d)$.

We focus on intervals $J_n$ for $n \geq 0$; similar arguments apply for $n \geq
0$. We consider two cases: $n<N$ and $n=N$.

\noindent
\emph{Case $n<N$.} It holds $\wid(J_n) = \rho_n = 3/2^{n+3}$. For each $J_n$,
assume that we perform a number of subdivision steps
to obtain intervals, say $J_{n, \ell}$, with
$\wid(J_{n,\ell}) = 2^{-\ell}$. We
choose $\ell$ so that the corresponding Obreshkoff areas,
$\ObA_{\varrho}(J_{n,\ell})$, are inside $\Omega_d$. In particular, we ensure
that the Obreshkoff areas related to $J_{n,\ell}$ lie in $D_{n+1}$.
  
The diameter of the Obreshkoff discs,
$\uObD_{\varrho}(J_{n,\ell})$ and $\dObD_{\varrho}(J_{n,\ell})$, 
is $\wid(J_{n,\ell})/\sin\tfrac{\pi}{\varrho + 2}$.
For every $\ObA_{\varrho}(J_{n,\ell})$ to be in $D_{n+1}$ and hence inside
$\Omega_d$,
it suffices that a disc with diameter 
$2 \,\wid(J_{n,\ell})/\sin\tfrac{\pi}{\varrho + 2}$,
that has its center in the interval $[\xi_n, \xi_{n+1}]$
and touches the right endpoint of $J_n$, to be
inside $D_{n+1} \setminus D_n$.
This is the worst case scenario: a disc big enough that contains
$\ObA_{\varrho}(J_{n,\ell})$ and lies $D_{n+1}$. This auxiliary disc is the
dotted (red) disc in Fig.~2~(left).
It should be that 
  \[
    2\, \wid(J_{n,\ell})/ {\sin\tfrac{\pi}{\varrho + 2}} \leq 
    2\,\rho_{n+1} = 3/2^{n+3}.
  \]
Taking into account that  $\wid(J_{n,\ell}) = 2^{-\ell}$ and
  \[
    \sin\tfrac{\pi}{\varrho + 2}
    > \sin\tfrac{1}{\varrho}
    \geq {\tfrac{1}{\varrho}} / {\sqrt{1 + \tfrac{1}{\varrho^2}}}
    \geq \tfrac{1}{2 \varrho} ,
    \]
we deduce
$
	2^{-\ell +1} 2 \varrho \leq 3/2^{n+3}
$
and so
$
	\ell \geq \log \frac{2^{n+5} \varrho}{3}.
$

Hence, $\wid(J_{n, \ell}) = 3/(2^{n+5} \varrho)$
and so $J_n$ is partitioned to at most
$\frac{\wid(J_n)}{\wid(J_{n,\ell})} = 4 \varrho$
(sub)intervals.
So, during the subdivision process, starting from (each) $J_n$, we obtain the
intervals
$J_{n,\ell}$ after performing
 at most $8 \varrho$ subdivision steps (this is the size of the
 complete binary tree starting from $J_n$). 
 To say it differently, the subdivision tree
 that has $J_n$ as its root and the intervals $J_{n, \ell}$ 
as leaves has depth $\ell = \lceil \log(4 \varrho) \rceil$. The same hold for
$J_{N-1}$ because
$\rho_n \leq \rho_{N}$, for all $0 \leq n \leq N-1$. 

Thus, the width of the tree starting at $J_n$ is at most $\Oh(\varrho^2)$,
because we have $\OO(\varrho)$ subintervals $J_{n,\ell}$
and for each $\var(f, J_{n,\ell}) \leq \varrho$.

\noindent
\emph{Case $n=N$.} 
Now $\wid(J_N) = 3/2^{N+1}$.
We need a slightly different argument to account for the number of subdivision
steps
for the last disc $D_N$. 
To this disc we assign the interval $J_N = [1 - 1/2^N, 1]$ with $\wid(J_N) =
1/2^N$;
see Figure~2.

We need to obtain small enough intervals $J_{N, \ell}$ of width $1/2^{\ell}$
so that corresponding Obreskoff areas,  $\ObA_{\varrho}(J_{N,\ell})$, to be
inside $D_N$.
So, we require that an auxiliary disc of diameter 
$2 \,\wid(J_{N,\ell})/\sin\tfrac{\pi}{\varrho + 2}$,
that has ts center in the interval $[1, 1/2^{N+1}]$
and touches 1 to be
inside $D_{N}$; actually inside $D_N \cap \{x \geq 1\}$; see Figure~2.
And so 
 \[
    2\, \wid(J_{N,\ell}) / {\sin\tfrac{\pi}{\varrho + 2}} \leq 
    \rho_{n+1} = 1/2^{N+1}.
  \]
This leads to $\ell \geq \log(\varrho \, 2^{N+3})$.
Working as previously, we estimate that the number of subdivisions we perform
to obtain the interval $J_{N, \ell}$ is $8\varrho$.
Also repeating the previous arguments, the width of the tree of \Descartes
starting at $J_N$ is at most $\Oh(\varrho^2)$.

By combining all the previous estimates, we conclude that the subdivision tree
of \descartes has width $\Oh(\varrho^2\log^2 d+\log^2 d)$.

To bound the depth of the subdivision tree of \descartes, consider an interval
$J_\ell$ of width $1/2^{\ell}$ obtained after $\ell+1$ subdivisions. By
theorem~\ref{thm:sepdepthbound}, we can guarantee termination if for some
$\varepsilon>0$,
\[
1/2^{\ell - 1}\leq \min\{\Delta_{\varepsilon}\uR(f),\varepsilon\}.
\]
Fix $\varepsilon=1/(\enumber d\condR(f))$. Then, by
Theorem~\ref{theo:condbasedseparation}, it suffices to hold
\[
\ell\geq 1+\log(12 d\condR(f)).
\]
Hence, the depth of the subdivision tree is at most $\Oh(\log(d\condR(f)))$.

Therefore, since the subdivision tree of \descartes has width
$\Oh(\varrho^2\log d+\log^2 d)$ and depth $\Oh(\log(d\condR(f)))$, the size
bound follows. For the bit complexity, by \cite{ESY:descartes}, see also
\cite{KM-newDesc-06,sm-anewdsc,Sagraloff-approxDesc-14,emt-lncs-2006} and
Proposition~\ref{prop:poly-maps},
the worst case cost of each step of \descartes is $\sOB(d \tau + d^2 \delta)$,
where $\delta$ is the logarithm of the highest bitsize that we compute with, or
equivalently the depth of the subdivision tree. In our case, $\delta =
\OO(\log(d \condR(f))$.
\end{proof}

\paragraph{Expected complexity estimates}

\begin{theo}
  \label{thm:Descartes-complexity-exp}
Let $\fkf \in \Pd^\bbZ$ be a random bit polynomial with $\tau(\fkf) \geq
\Omega(\log{d} +u(\fkf))$. Then, using \descartes, the expected number of
subdivision steps to isolate the real roots in $I = (-1, 1)$	is
	\[
	\sO((1+u(\fkf))^3).
	\] 
	The expected bit complexity of \descartes is
  \[
    \sOB(d \, \tau(\fkf) (1+u(\fkf))^3 + d^2(1+u(\fkf))^4 ).
  \]
If $\fkf$ is a uniform random bit polynomial of bitsize $\tau$ and $\tau =
\Omega(\log{d} +u(\fkf))$, then the expected number of subdivision steps to
isolate the real roots in $I = (-1, 1)$	is
	$
	\sO(1)
	$ 
and the expected bit  complexity becomes
  \[
    \sOB(d \tau  + d^2 ).
  \]
\end{theo}
\begin{proof}
We only bound the number of bit operations; the bound for the number of steps
is analogous. By Theorem~\ref{thm:Descartes-steps} and the worst-case bound
$\sOB(d^4 \tau^2)$ for \descartes \cite{ESY:descartes}, the bit complexity of
\descartes at $\fkf$ is at most
\[
\sOB\left(\min\{d \tau(\fkf) \varrho(\fkf)^2 \log \condR(\fkf)+ d^2
\varrho(\fkf)^2 \log^2\condR(\fkf)),d^4\tau(\fkf)^2\}\right),
\]
that in turn we can bound by 
\begin{equation*}
\sOB\left(d \tau(\fkf)\varrho(\fkf)^2
\min\{\log \condR(\fkf),d^3\tau(\fkf)\}  +  
d^2 \varrho(\fkf)^2\min \{\log\condR(\fkf),d^2\tau(\fkf)^2 )\}\right).
\end{equation*}
Now, we take expectations, and, by linearity, we only need to bound
\[
\bbE\,\varrho(\fkf)^2 \min\{\log \condR(\fkf),d^3\tau(\fkf)\}
\quad \text{ and } \quad
\bbE\,\varrho(\fkf)^2 \left(\min \{\log\condR(\fkf),d^2\tau(\fkf)^2\}\right)^2 .
\]
Let us show how to bound the first, because the second one is the same. By the
Cauchy-Bunyakovsky-Schwarz inequality, $$\bbE\,\varrho(\fkf)^2 \min\{\log
\condR(\fkf),d^3\tau(\fkf)\}$$ is bounded by
\[
\sqrt{\bbE\,\varrho(\fkf)^4}\sqrt{\bbE\,\left(\min\{\log
\condR(\fkf),d^3\tau(\fkf)\}\right)^2}.
\]
Finally, Corollaries \ref{cor:realglobalexpectations} and
\ref{cor:probboundroots} give the estimate.
Note that $\tau(\fkf) \geq \Omega(\log{d} +u(\fkf))$ implies $\tau(\fkf) \geq
\Omega(\log{d} +u(\fkf)+\ln c)$ (for the worst-case separation bound $c$
\cite{Dav:TR:85}) so we can apply Corollary~\ref{cor:realglobalexpectations}.
\end{proof}

\subsection{\textsc{Sturm} solver}
\label{sec:Sturm}

\sturm solvers is based on (evaluations of) the Sturm
sequence of $f$ to count the number of real roots, say $\varrho$, of a polynomial in an interval, in our case $I = [-1, 1]$.

Given a real univariate polynomial $f$ of degree $d$, and its derivative $f'$, the Sturm sequence of $f$ 
is a sequence of polynomials $F_0, F_1, \dots$, such that 
$F_0 = f$, $F_1 =f'$, and $F_i = - \rem(F_{i-2}, F_{i-1})$, for $i \geq 2$. 
We denote this sequence as $\ST(f)$.
Notice that the sequence contains at most $d+1$ polynomials and the degree of $F_i$ is at most $d -i$; hence there are in total $\OO(d^2)$ coefficients
in the sequence.

If $a \in \RR$, then $\ST(f; a) := \{F_0(a), F_1(a), F_2(a), \dots \}$ is the evaluation of the polynomials in the Sturm sequence at $a$.
Also, we denote the number of sign variations (zeros excluded) in this sequence as $\var(\ST(f; a))$.
Sturm's theorem states that the number of distinct real roots of $f$ in an interval $[a, b]$ is $\var(\ST(f; a)) - \var(\ST(f; b))$.
	We exclude the cases where $f(a) = 0$ or $f(b) = 0$, as we can treat them, easily, independently.
Sturm's theorem does not assume that $f$ is square-free and it counts exactly the number of real roots of a polynomial in an interval. Thus, it is straightforward to come up with a subdivision algorithm, based on Sturm's theorem, to isolate the real roots of $f$; this is the so-called \sturm solver that mimics, in a precise way, the binary search algorithm.

The pseudo-code of \sturm (Alg.~\ref{alg:Sturm}) is almost the same with the pseuso-code of 
\descartes algorithm. They only differ at Line 4, which represents the way that we count the real roots of a polynomial
in an interval. \sturm counts exactly using Sturm's sequences, 
while \descartes provides an upper bound on the number of real roots using the Descartes' rule of signs.

\begin{algorithm2e}[t]
  \scriptsize \dontprintsemicolon \linesnumbered
  \SetFuncSty{textsc} \SetKw{RET}{{\sc return}} \SetKw{OUT}{{\sc output \ }}
  \SetVline \KwIn{A square-free polynomial $f\in \Pd^\bbZ$}
\KwOut{A list, $S$, of isolating intervals for the real roots of $f$ in $J_0 =
(-1, 1)$}

  \BlankLine

  $J_0 \leftarrow (-1, 1)$, 
  $S \leftarrow \emptyset,\, Q \leftarrow \emptyset$,
  $Q \leftarrow \FuncSty{push}( {J_0})$ \;

  \While{ $Q \neq \emptyset$}{
    \nllabel{alg:Subdivision-while-loop}

    ${J}=(a,b) \leftarrow \FuncSty{pop}( Q)$ \\
    $V \leftarrow \var(\ST(f; a)) - \var(\ST(f; b))$ \;

    \Switch{ $V$ }{

      \lCase{ $V = 0$ }{ \KwSty{continue}\; }
      \lCase{ $V = 1$ }{ $S \leftarrow \FuncSty{ add}( {I})$ \; }
      \Case{ $V > 1$ } {
        $m \gets \frac{a+b}{2}$ \;
        \lIf{$f(m) = 0$}{  $S \leftarrow \FuncSty{ add}( {[m, m]})$ \; }
        $J_L \gets [a, m]$ ;   $J_R \gets [m, b]$ \;

        $Q \leftarrow \FuncSty{push}( Q, {J_L} )$,
        $Q \leftarrow \FuncSty{push}( Q, {J_R} )$ \;
      }
    }
  }
  \RET $S$ \;
  \caption{ $\func{Sturm}(f)$}
  \label{alg:Sturm}
\end{algorithm2e}

\sturm isolates  the real roots of a polynomial $f$ with integer coefficients in $I$. Suppose there are $\varrho$ many roots, and note that we only evaluate $\ST(f)$ on rational numbers in \sturm implementation. Now we consider the complexity the evaluation step: Most, if not all, the implementations of \sturm represent and evaluate a Sturm sequence straightforwardly.  That is, they compute all the polynomials in $\ST(f)$ and then evaluate them at various rational numbers.
 There are at most $d+1$ polynomials in the sequence, having degree at most $d-i$. Hence, there are $\OO(d^2)$ coefficients having worst case bitsize $\sO(d \tau)$
 \cite{vzGGer}. Thus, their total bitsize is $\sO(d^3 \tau)$.

 A faster approach to evaluating Sturm sequence is provided by ``half-gcd'' algorithm \cite{Reischert:subresultant:97}.
In ``half-gcd'' approach we essentially exploit the polynomial division relation
$F_{i-2} = Q_i F_{i-1} - (-F_i)$: We notice that, using this relation, the evaluations $F_{i-2}$ and $F_{i-1}$ at $a$, and the evaluation of  the quotient $Q_i$
suffices to compute $F_i(a)$. Thus, initially, we evaluate the polynomials $F_1 := f$ and $F_2 := f'$, in $\sOB(d (\sigma + \tau))$, and then, using the sequence of quotients we compute the evaluation of the sequence.  
There are at most $\sO(d)$ quotients in the sequence, having in total $\sO(d)$ coefficients, 
of (worst case) bitsize $\sO(d \tau)$ \cite{Reischert:subresultant:97}.
In this way we can evaluate the whole Sturm sequence at a number of bitsize $\sigma$ with complexity $\sOB(d^2 (\sigma  + \tau))$ \cite{Reischert:subresultant:97},


The following proposition demonstrates the worst case bit complexity assuming the "half-gcd" approach to pointwise evaluation of Sturm sequence. 
The proof is not new, but we modify it to express the complexity as function of the real condition number.
We refer the
reader to \cite{Yap:SturmBound:05,Dav:TR:85,emt-lncs-2006} and references
therein for further details.

\begin{lem}
  \label{lem:Sturm-complexity-C}
  Let $f \in \Pd^\bbZ$ of bitsize $\tau$.
  The bit complexity of \sturm to isolate the real roots of $f$ in $I$, say there are $\varrho$, is
  \[
  \sOB( \varrho d^2 \delta (\tau + \delta)),
  \]
  where $\delta$ is the bitsize of the separation bound
  of the root of $f$, or
  \[
    \sOB(\varrho \, d^2 \, \log{\condR(f)} ( \tau + \log{\condR(f)})), 
  \]
  where $\condR(f)$ is the global condition number of $f$, see \eqref{eq:condR}.
\end{lem}
\begin{proof}
Let $\varepsilon=0$ and $\varrho$ the number of roots of $f$ in $I_0=I$.

  Let $\Delta_j$ be the (real) \textit{local separation} bound of the real roots, say $\alpha_j$,
  of $f$ in $I$; that is 
  \[
    \Delta_j = \Delta(f, \alpha_j) = \min\nolimits_{i \not= j}\abs{\alpha_i - \alpha_j};
  \]
 also let $\Delta = \min_{j \in [\varrho]}\Delta_j$ and $\delta = -\log\Delta$.
  
 To isolate the real roots in $I$ we need to compute $\varrho -1$ 
rational numbers between them. As \sturm mimics binary search, the resulting intervals
have width at least $\tfrac{\Delta_j}{2}$
and the number of subdivision steps we need to perform is at  $\lceil \log{ \frac{4}{\Delta_j}} \rceil$, for $1 \leq j \leq \varrho$.
  Let $T$ be the binary tree corresponding to the realization of \sturm
  and let  $\#( T)$ be the number of its nodes;
  or in other words the \emph{total} number of subdivisions that \sturm performs.
  Then
  \begin{equation}
    \#(T) =  \sum_{j=1}^{\varrho}{ \lceil{ \log{ \frac{4}{\Delta_j}}} \rceil}
    \le 3 \varrho - \sum_{j=1}^{\varrho}{ \log{\Delta_j}}
    = 3 \varrho - \log \prod_{i= 1}^{\varrho}\Delta_i
    \leq \varrho (3 - \log\Delta) = \varrho (3 + \delta).
    \label{eq:sturm-nb-of-steps}
  \end{equation}
  The complexity of \sturm algorithm is the number of step it performs,
  $\#(T)$, times the worst case (bit) complexity of each step.  Each
  step corresponds to an evaluation of the Sturm sequence at a number.
  If the bitsize of this number is $\sigma$, then the cost is
  $\sOB(d^2 \, (\tau + \sigma))$ \cite{Reischert:subresultant:97}.
  In our case, $\sigma = 3 -\log\Delta= 3 + \delta$.
  Therefore, the overall cost is
  \begin{equation*}
    \label{eq:sturm-all-roots-complexity}
    \sO(\varrho \delta) \cdot \sOB(d^2 \, (\tau + \delta)) = \sOB( \varrho d^2 \delta (\tau + \delta)).
  \end{equation*}
  To obtain the complexity bound involving the condition number,
  we notice that Theorem~\ref{theo:condbasedseparation} implies 
  $\delta = \OO( \log( d \condR(f)) )$.
\end{proof}

\begin{remark}
	\label{rem:sturm-worst-case}
	The standard approach to analysis of  \sturm  relies on aggregate separation bounds, e.g., \cite{emt-dmm-j-19}; this approach yields a bound of the order $\sOB(d^4 \tau^2)$.
\end{remark}

\begin{theo}
  \label{thm:Sturm-complexity-exp}
  Let $\fkf\in\Pd^\bbZ$ be a random bit polynomial of bit-size $\tau(\fkf) \geq 10$, and
  uniformity $u(\fkf)$ (Def.~\ref{def:uniformity}).
  If $\tau(\fkf) = \Omega( \log{d}+ u(\fkf))$,
  then the
  expected bit complexity of \sturm to isolate the real roots of $\fkf$ in $I=[-1, 1]$,
  \emph{using fast algorithms for evaluating Sturm sequences}, is
  $\sOB(d^2 \tau(\fkf) \, (1 + u(\fkf))^3)$.

   If $\fkf$ has uniformly distributed
  coefficients on $[-2^{\tau}, 2^{\tau}] \, \cap \, \mathbb{Z}$, then
  the complexity is  $\sOB(d^2 \tau)$.
\end{theo}
\begin{proof}
  Assume that $\fkf \in\Pd^\bbZ$ is a random bit polynomial of bit-size $\tau = \tau(\fkf)$, not necessarily square-free.
  Using \sturm, the worst case complexity for isolating its real roots in $I$ is $\sOB(d^4 \tau^2)$ \cite{Yap:SturmBound:05},
  while 
  Lemma~\ref{lem:Sturm-complexity-C} implies the bound $\sOB(\varrho d^2 \, \log{\condR(\fkf)} ( \tau + \log{\condR(\fkf)}))$.
  Thus the complexity is
  \begin{equation}
    \label{eq:Sturm-bound-min}
    \min\{ \sOB(\varrho d^2 \, \log{\condR(\fkf)} ( \tau + \log{\condR(\fkf)})), \sOB(d^4 \tau^2) \}
    = \sOB(d^2 \tau) \, \varrho \, \min\{ \log^2{\condR(\fkf)}, d^2 \tau \}.
  \end{equation}

  For the random bit polynomial $\fkf$, with $\tau(\fkf) \geq 4\log(ed)+2u(\fkf)+ 12 \log(d \,\tau(\fkf))$,
  which for $\tau(\fkf) \geq 10$ becomes 
  $\tau(\fkf) = \Omega( \log{d}+ u(\fkf))$,
  using Cor.~\ref{cor:realglobalexpectations} with $\ell = 2$
  we get
  \[
    \bbE_\fkf \left( \min\{\ln\condR(\fkf), d^2\tau(\fkf) \}\right)^2  = \OO((\log{d} + u(\fkf))^2).
  \]
  
  Corollary~\ref{cor:probboundroots}, using the same constraints on $\tau(\fkf)$   and $\ell = 1$, implies that $\bbE_\fkf (\varrho) = \OO(\log{d}(\log{d} + u(\fkf)))$.
  Notice that we implicitly assume that the (random variables) $\varrho$ and $\condR(f)$
	are independent.
Combining all the previous estimates, we deduce that the 
expected runtime of \sturm for $\fkf$ is
  $\sOB(d^2 \tau(\fkf) (1 + u(\fkf)^3))$.
\end{proof}

With the standard representation of Sturm sequence, we evaluate $\ST(f)$ at a rational number of bitsize $\sigma$ in $\sOB(d (d^2 \sigma + d^2 \tau)$. As we have to perform this evaluation $\varrho$ times, the total complexity is
$\sOB(\varrho (d^3 \sigma + d^3 \tau))$. 
This is worse than the bound for evaluation used in the proof of Theorem~\ref{thm:Sturm-complexity-exp},
which was  $\sOB(\varrho \,d^2 \, (\tau + \sigma))$,
by a factor of $d$. To obtain the worst case bound for \sturm with this representation it suffices to replace $\sigma$ with $\delta$, respectively $\log{\condR(f)}$, 
to obtain 
$\sOB(\varrho (d^3 \delta + d^3 \tau))$,
respectively 
$\sOB(\varrho (d^3 \log{\condR(f)} + d^3 \tau))$. 


In practice, \sturm is rarely used. It is slower that \descartes by several orders of magnitude, almost always, e.g.~\cite{htzekm-solve-09}.  We give a theoretical justification of these practical observations. The following "assumption" corresponds to the current status of all implementations of the \sturm algorithm to the authors' knowledge.

\begin{assumption}
	\label{ass:Sturm-representation}
	We assume that we represent Sturm sequence of a polynomial $f$ of degree $d$ and bitsize $\tau$, as $\ST(F) = \{F_0, F_1, \dots, \}$, where 
	$F_1 =f'$, and $F_i = - \rem(F_{i-2}, F_{i-1})$, for $i \geq 2$. 
\end{assumption}

\begin{prop}   \label{lem:Sturm-lower-bound}
Let $f \in \Pd^\bbZ$ of bitsize $\tau$.
Under the Assumption~\ref{ass:Sturm-representation}, the expected complexity of \sturm for a random bit polynomial of bit-size $\tau$ is $\Omega(d^3 + d^2 \tau )$.
\end{prop}
\begin{proof}
The bitsize of coefficients in the sequence is $\Omega(\tau)$. Thus, under Assumption~\ref{ass:Sturm-representation}, the overall complexity of the algorithm becomes  $ \Omega (\varrho (d^3 \sigma + d^2\tau)$.
This implies that, independently
of the bounds on $\varrho$ and $\sigma$, 
a lower bound on the complexity of \sturm is
$\Omega(d^3 + d^2\tau)$.
\end{proof}

We believe this simple proposition compared to \Cref{thm:Descartes-complexity-exp} explains the practical superiority of \descartes over current implementations of \sturm.

\begin{remark}
A natural question is to ask for a lower bound in the case \sturm is implemented using ``half-gcd'' approach. In this case, one can set-up the ``half-gcd'' computation as a martingales and analyze its bit-complexity. Since only evaluating the beginning of the sequence costs $\OO(d \tau)$ bits, this approach is likely to yield a lower bound that still separates \sturm from the upper bound obtained for \descartes in  \Cref{thm:Descartes-complexity-exp}. We refrain from performing this analysis for the sake of not adding more technicality to our paper.
\end{remark}




\subsection{\textsc{ANewDsc}}
\label{sec:aNewDsc}

Sagraloff and Merhlhorn \cite{sm-anewdsc} presented an algorithm, \anewdsc, 
to isolate the real roots of a square-free univariate polynomial $f$
that combines \descartes with Newton iterations. 
If $f$ is of degree $d$, its roots are $\alpha_i$, for $i \in [d]$, and its leading coefficient is in the interval $(\tfrac{1}{4}, 1]$, then the bit complexity of the algorithm
is 
\[ \sOB\left(d (d^2 + d\log\mathcal{M}(f) + \sum\nolimits_{i=1}^{d} \log 1/f'(\alpha_i)) \right) ,
\]
where $f'$ is the derivative of $f$
and $\mathcal{M}(f)$ is the Mahler measure of $f$; it holds $\mathcal{M}(f) \leq \norm{f}_2$ \cite[Lem~4.14]{yap-fca-00}.
If the bitsize of $f$ is bounded by $\tau$, then the bound of the algorithm becomes $\sOB(d^3 + d^2 \tau)$.

However, if we are interested in isolating the real roots of $f$ in an interval, say $I$, then 
only the roots that are in the complex disc that has $I$ as a diameter affect the complexity bound. 
Therefore, if these roots are at most $\rho$, 
the first $d^3$ summand in the complexity bound becomes $d^2 \rho$;
moreover, we should account for the evaluation of the derivative of $f$ only at these roots. 
Regarding the evaluation of $f'$ over the roots of $f$, it holds
\[	
	\abs{f'(\alpha_i)} = a_d \prod_{j \not=i}\abs{\alpha_i - \alpha_j} \geq 
	a_d \Delta_i^{d-1}
	\Rightarrow
	- \log \abs{f'(\alpha_i)} \leq - (d-1) \log{\Delta_i}.
\]

Using these observations, and by also considering
 $\Delta = \min_{j \in [\varrho]}\Delta_j$ and $\delta = -\log\Delta$ the complexity bound becomes
\[ \sOB\left(d^2 \varrho + d \varrho \log\mathcal{M}(f) +  d^2 \varrho \delta) \right) =
\sOB\left(\varrho( d^2 + d \tau +  d^2 \log(\condR(f))) \right).
\]

\begin{theo}
  \label{thm:aNewDsc-complexity-exp}
Let $\fkf \in \Pd^\bbZ$ be a random bit polynomial with $\tau(\fkf) \geq
\Omega(\log{d} +u(\fkf))$. Then, the expected bit complexity of \anewdsc is
  \[
    \sOB((d^2 + d \, \tau(\fkf)) (1+u(\fkf))^2 ).
  \]
If $\fkf$ is a uniform random bit polynomial of bitsize $\tau$ and $\tau =
\Omega(\log{d} +u(\fkf))$, then the expected bit  complexity becomes
  \[
    \sOB(d^2 + d \tau ).
  \]
\end{theo}
\begin{proof}
We only bound the number of bit operations; the bound for the number of steps
is analogous. The worst-case bound
$\sOB(d^3 + d^2 \tau)$. Thus the bit complexity of \anewdsc at $\fkf$ is at most
\[
\sOB\left(\min\{d^3 +d^2 \tau(\fkf), \varrho( d^2 + d \tau(\fkf) +  d^2 \log(\condR(f))) \}\right) = 
\sOB\left(\varrho( d^2 + d \tau(\fkf) +  d^2 \log(\condR(f))) \right).
\]
Now, we take expectations, and, by linearity, we only need to bound
\[
\bbE\,\varrho(\fkf) \tau(\fkf)
\quad \text{ and } \quad
\bbE\,\varrho(\fkf) \log \condR(\fkf)  .
\]
For the random bit polynomial $\fkf$, with 
 $\tau(\fkf)\geq 12 \ln(\enumber d)+2u(\fkf) = \Omega( \log{d}+ u(\fkf))$,
 using Corollary~\ref{cor:probboundroots},
 we have  
 $\bbE_\fkf (\varrho) = \OO(\log{d}(1 + u(\fkf)))$.
 
To bound the other expectation, we use Cauchy-Bunyakovsky-Schwarz inequality, 
that is 
\[
\bbE\,\varrho(\fkf) \log \condR(\fkf) 
\leq \sqrt{ \bbE\,\varrho(\fkf)^2}  \sqrt{\bbE\log \condR(\fkf)^2} 
\]
Using again Corollary~\ref{cor:probboundroots}, with $\ell =2$, 
we have that 
$\sqrt{ \bbE_\fkf (\varrho)^2} = \OO(\log{d}(1 + u(\fkf)))$.
Similarly, using 
Corollary~\ref{cor:realglobalexpectations} 
\[ \sqrt{\bbE\log \condR(\fkf)^2}  =\OO(\log{d} + u(\fkf)). \]

Combining all the previous bounds, we arrive at the announced bound.
\end{proof}
\subsection{\jssparse algorithm by Jindal and Sagraloff}
\label{sec:jindal-sagraloff}

An important, both from a theoretical and a practical point of view, variant of the (real) root isolation problem 
is the formulation that accounts for sparsity of the input equation. In this setting, the input consists of (i) the non-zero coefficients, let their set (or support) be $M$ and their number be $\abs{M}$,\
(ii) the bitsize of the polynomial, say it is $\tau$, 
and (iii) the degree of the polynomial, say $d$, However, in this  sparse encoding, we need $\OO(\log{d})$ bits to represent the degree.  
Thus, the input is of bitsize $\sO(\abs{M} \tau \log(d))$;
we call this the \emph{sparse encoding}.
In the dense case $\abs{M} = d$ and the input has bitsize $\OO(d \tau)$.

As already mentioned, in the worst case, the bitsize of the separation bound is $\log\Delta = \sO(d \tau)$.
This result rules out the existence of a polynomial time, with respect to sparse encoding, algorithm for root isolation. 
The current state-of-art algorithm by Jindal and Sagraloff \cite{JinSag-sparse-17},
we call it \jssparse.
It has bit complexity polynomial in quantities $\abs{M}, \tau$, and $\log\Delta$.
Using Theorem~\ref{theo:condbasedseparation}
we can express the complexity bound of \jssparse using the condition number of the polynomial. In particular:
\begin{prop}
\label{prop:JScomplexitycondition}
Given $f \in\Pd$ with support $|M|$, \jssparse  computes isolating intervals for all the roots of $f$ in $I$ by performing 
\[
\OB\left(|M|^{12}\log^{3}d \max\{\log^2\|f\|_1,\log^3 \condR(f)\}\right)
\]
bit operations.
\end{prop}

Even though the worst case bound of \jssparse is exponential with respect to the sparse encoding, it is the fist algorithm that actually depends on the actual separation bound of the input polynomial and exploits the support. 

In out probabilistic setting, the following result is immediate 

\begin{theo}
	\label{thm:jssparse-complexity-exp}
If $\fkf$ is a uniform random bit polynomial of bitsize $\tau$ and $\tau =
\Omega(\log{d} +u(\fkf))$, 
having support $|M|$, then  \jssparse computes isolating intervals for all the roots of $f$ in $I$ 
    in expected bit complexity 
    \[
\sOB\left(|M|^{12} \, \tau^2  \, \log^{3}{d} \right)
\]
under the (reasonable) assumption that $\tau > \log^3{d}$.
\end{theo}
\label{sed:conclusion}

\begin{acks}
J.T-C. was partially supported by a postdoctoral fellowship of the 2020
``Interaction'' program of the Fondation Sciences Mathématiques de Paris, and some funds from 2023 AMS-Simons Travel Gran during the writing of this paper. He is
grateful to Evgenia Lagoda for emotional support during the thinking period,  Brittany Shannahan and Lewie-Napoleon III for
emotional support during the writing period and Jazz G. Suchen for useful
suggestions regarding
Proposition~\ref{prop:reallinearprojectiondiscreterandomvector}.
A.E. was partially supported by NSF CCF 2110075 and NSF CCF 2414160, J.T-C. and E.T. 
were partially
supported by ANR
JCJC GALOP (ANR-17-CE40-0009).
\end{acks}

\bibliographystyle{ACM-Reference-Format}
\bibliography{BIBLIO.bib}


\end{document}